\newcommand*{\QUANTUM}{}%
\ifdefined\QUANTUM
\documentclass[a4paper,onecolumn,11pt,accepted=2023-05-29]{quantumarticle}
\pdfoutput=1
\else
\documentclass[11pt]{article}
\fi

\usepackage{amsmath,amssymb,amsthm}
\usepackage{latexsym}
\usepackage{indentfirst}
\usepackage{placeins}
\usepackage{booktabs}
\usepackage{algorithm}
\usepackage{algorithmic}
\usepackage{multirow}
\usepackage{color, xcolor, colortbl}
\usepackage{subfigure}
\usepackage{qcircuit}

\usepackage{graphicx}
\usepackage{bm}
\usepackage{braket}

\usepackage[sort,square,numbers]{natbib}
\usepackage[T1]{fontenc}
\usepackage{hyperref}
\usepackage[capitalize]{cleveref}
\makeatletter

\theoremstyle{plain}
\newtheorem{theorem}{Theorem}
\newtheorem{lemma}[theorem]{Lemma}
\newtheorem{prop}[theorem]{Proposition}
\newtheorem{cor}[theorem]{Corollary}
\theoremstyle{definition}
\newtheorem{definition}{Definition}

\theoremstyle{definition}

\theoremstyle{remark}
\newtheorem{remark}{Remark}

\DeclareMathOperator{\polylog}{polylog}
\DeclareMathOperator{\poly}{poly}

\newcommand{\abs}[1]{\left|#1\right|}
\newcommand{\norm}[1]{\lVert#1\rVert}
\newcommand{\grad}{\nabla}
\newcommand{\z}{{\bm{\xi}}}
\newcommand{\x}{{\bm{x}}}
\newcommand{\R}{\mathbb{R}}

\newcommand{\q}{{\bm{q}}}

\newcommand{\xp}{\frac{\x}{P}}
\newcommand{\xkp}{\frac{x_k}{P}}
\renewcommand{\a}{{\breve{a}}}
\newcommand{\ba}{{\breve{\alpha}}}
\newcommand{\bb}{{\breve{\beta}}}
\renewcommand{\t}{\tau}
\newcommand{\Z}{\Xi}

\newcommand{\diag}{\mathrm{diag}~}

\newcommand{\half}{\frac{1}{2}}
\newcommand{\mo}{\mathcal{O}}
\newcommand{\de}[1]{\deg_a\left(#1\right)}
\newcommand{\dg}[3]{\eta_{#1}\left(#2, #3\right)}

\newcommand{\eps}{\epsilon}

\begin{document}

\title{On efficient quantum block encoding of pseudo-differential operators}

\author{Haoya Li}
\affiliation{Department of Mathematics, Stanford University, Stanford, CA 94305} 
\orcid{0000-0001-7076-7600}
\email{lihaoya@stanford.edu}

\author{Hongkang Ni} 
\affiliation{Institute for Computational and Mathematical Engineering, Stanford University,  Stanford, CA 94305} 
\orcid{0000-0002-7507-4755}
\email{hongkang@stanford.edu}

\author{Lexing Ying} 
\affiliation{Department of Mathematics, Stanford University, Stanford, CA 94305}
\affiliation{Institute for Computational and Mathematical Engineering, Stanford University,  Stanford, CA 94305} 
\orcid{0000-0003-1547-1457}
\email{lexing@stanford.edu}

\begin{abstract}

Block encoding lies at the core of many existing quantum algorithms. Meanwhile, efficient and explicit block encodings of dense operators are commonly acknowledged as a challenging problem. This paper presents a comprehensive study of the block encoding of a rich family of dense operators: the pseudo-differential operators (PDOs). First, a block encoding scheme for generic PDOs is developed. Then we propose a more efficient scheme for PDOs with a separable structure. Finally, we demonstrate an explicit and efficient block encoding algorithm for PDOs with a dimension-wise fully separable structure. Complexity analysis is provided for all block encoding algorithms presented. The application of theoretical results is illustrated with worked examples, including the representation of variable coefficient elliptic operators and the computation of the inverse of elliptic operators without invoking quantum linear system algorithms (QLSAs). 
\end{abstract}

\maketitle

\section{Introduction}\label{sec:intro}

Block encoding \cite{low2017optimal} is a widely used technique in quantum computing and a crucial component of many quantum algorithms with a potentially exponential advantage over classical algorithms, such as quantum phase estimation (QPE) \cite{kitaev2002classical, nielsen2002quantum}, the HHL algorithm \cite{harrow2009quantum}, quantum singular value transformation (QSVT) \cite{gilyen2019quantum} and various quantum linear system solvers \cite{lin2020optimal, an2022quantum, costa2022optimal}, to name a few. The idea of block encoding is to embed a linear operator $A$ into a unitary operator $U_A$ with larger dimensions after appropriate scaling. The unitary $U_A$ is then converted into a quantum circuit, allowing a quantum computer to access $U_A$ for actual computations. 

The potential advantage of quantum algorithms depends critically on efficient and practical quantum circuits for block-encoding of the operators involved, and the construction of such circuits can be non-trivial in general. Researchers have constructed block encoding schemes leveraging different structures of the operators studied. 
For example, a block encoding scheme is provided in \cite{camps2022explicit, gilyen2019quantum} for sparse matrices, and a recipe is presented for hierarchical matrices in \cite{nguyen2022quantum}. 
In this paper, we consider the problem of block encoding a large family of dense operators: the pseudo-differential operators (PDOs). PDO is a rich family of linear operators that include many commonly used examples in scientific problems, which is typically given in the following form:
\begin{equation}\label{eq:pdo}
    A f(\x) =\int_{\mathbb{R}^d} e^{2 \pi i \x \cdot \z} a(\x, \z) \widehat{f}(\z) d \z, 
\end{equation}
where $a(\x, \z)\in C^\infty(\mathbb{R}^d\times\mathbb{R}^d)$ is called the symbol of $A$ and $\widehat{f}$ is the Fourier transform of $f$. A major motivation for studying operators with the form \eqref{eq:pdo} is that differential operators often enjoy a simple representation in the Fourier domain. For example, the elliptic operator:
\begin{equation}\label{eq:ellip}
A = I - \grad\cdot(\omega(\x)\grad)
\end{equation}
with $\omega(\x)>0$ can be represented in the form of \eqref{eq:pdo} with the symbol 
\begin{equation}\label{eq:ellipsym}
a(\x, \z) = 1-2\pi i\grad\omega(\x)\cdot\z+4\pi^2\omega(\x)|\z|^2. 
\end{equation}
More generally, an $m$-th order linear partial differential operator $P(\x, D) = \sum_{|\alpha|\leq m} a_\alpha(\x)D^\alpha$ with $D = -\frac{i}{2\pi}\grad_\x$ can be represented by
\[
Pf(\x) = \int_{\mathbb{R}^d} e^{2 \pi i \x \cdot \z} \sum_{|\alpha|\leq m} a_\alpha(\x) \z^\alpha \widehat{f}(\z) d \z = \int_{\mathbb{R}^d} e^{2 \pi i \x \cdot \z} P(\x, \z) \widehat{f}(\z) d \z,
\]
where $\alpha=(\alpha_1, \ldots, \alpha_d)$ is the $d$-dimensional multi-index and $|\alpha|=\sum_{j=1}^d\alpha_j$. 
Another popular example is the translation-invariant operator. Let $\varphi_\z(\x) = e^{2\pi i\x\cdot\z}$ be a function of $\x$. If an operator $A$ is translation-invariant, i.e., $(A\varphi_\z)(\x)=a(\z)\varphi_\z(\x)$, then 
\begin{equation}\label{eq:multiplier}
Af(\x) = \int_{\mathbb{R}^d} e^{2 \pi i \x \cdot \z}a(\z)\widehat{f}(\z) d \z,
\end{equation}
in which case we say that the symbol $a(\z)$ is a multiplier. Apart from the examples mentioned above, the PDO family also contains other operators such as convolution operators, singular integral operators, etc. Moreover, a space of PDOs is often closed with respect to many elementary operations under certain conditions. For example, for the operator $A$ in \eqref{eq:multiplier} with symbol $a(\z)\not=0$, the inverse of $A$ can be simply represented by
\[
A^{-1}f(\x)= \int_{\mathbb{R}^d} e^{2 \pi i \x \cdot \z} \frac{1}{a(\z)}\widehat{f}(\z) d \z.
\]
In general, the operator defined by $C^\infty$ function $a(\x,\z)$ as in \eqref{eq:pdo} is called a pseudo-differential operator only if $a(\x,\z)$ satisfies some additional requirements such as
\[
|\partial_\x^\alpha\partial_\z^\beta a(\x, \z)|\leq C_{\alpha\beta}\langle\z\rangle^{m-\alpha-\beta},
\]
where $\langle\z\rangle:=\sqrt{1+|\z|^2}$, and the space of the corresponding PDOs is denoted by $S^m$. There are multiple monographs on PDOs that interested readers can refer to, such as \cite{stein1993harmonic, wong1991anitro}. 

The PDOs considered in this paper are equipped with a periodic boundary condition on the space domain $\Omega = [0, 1]^d$. 
The frequency variable $\z$ thus takes the value on the integer grid, and the operator $A$ becomes
\begin{equation}\label{eq:pdoZZ}
    A f(\x) = \sum_{\z \in \mathbb{Z}^d} e^{2 \pi i \x \cdot \z} a(\x,\z) \hat{f}(\z),
\end{equation}
where $\hat{f}$ is the coefficient of the Fourier series of $f$. In this paper, we derive block encoding schemes for the PDO \eqref{eq:pdoZZ} based on different additional structures of the symbol $a(\x, \z)$.
First, we present a block encoding scheme for generic symbols $a(\x, \z)$ without additional structures. We then point out that the success probability of the quantum circuit can be significantly improved if the symbol $a(\x, \z)$ can be expanded into series:
\begin{equation}
    a(\x, \z) = \sum_j \alpha_j(\x)\beta_j(\z), 
\end{equation}
with only $\mo(1)$ terms. Furthermore, the circuit can be constructed in a much more explicit way with the help of quantum signal processing (QSP) and quantum eigenvalue transformation (QET) if the symbol is a sum of \emph{fully separable} terms, i.e., it can be expanded as
\begin{equation}\label{eq:fullysepex}
    a(\x, \z) = \sum_j \alpha_{j1}(x_1)\cdots\alpha_{jd}(x_d)\beta_{j1}(\xi_1)\cdots\beta_{jd}(\xi_d), 
\end{equation}
with $\mo(1)$ terms, where $\x=(x_1,\ldots,x_d)$ and $\z=(\xi_1,\ldots,\xi_d)$. See \Cref{def:fullysep} and \Cref{sec:fully sep} for details. Complexity analysis is included for all block encoding schemes, and their applications are showcased with specific examples. The contributions of this paper can be summarized as follows:
\begin{itemize}
    \item We provide practical block encoding schemes for pseudo-differential operators, including algorithms applicable to generic PDOs (see \Cref{fig:PDO1}), efficient block encoding for separable PDOs (see \Cref{fig:PDO2}) and explicit circuits for fully separable PDOs (see \Cref{fig:PDO3}). Novel ideas of circuit design, such as the phase multiplication circuit (see \Cref{fig:lemma2}) and the prototype for diagonal multiplication (see \Cref{fig:D}), are included in the block encoding schemes.
    \item We conduct comprehensive complexity analysis for the block encoding schemes proposed. The result for complexity analysis includes the success probability, the number of ancilla qubits needed, and the number of gates used. In addition to theorems applicable to general cases, we also demonstrate possibilities of improving the complexity results by leveraging particular structures of the problem (see \Cref{subsec:varellip} for example).
    \item We demonstrate the usage of our results with explicit examples. One can use our block-encoding scheme not only as an integrated part of established quantum algorithms but also as an option for conducting operations directly on certain operators. For the example shown in \Cref{subsec:appinverse}, we use the idea of symbol calculus to directly block-encode the inverse of an elliptic operator and the dependence of the complexity on $P$ (the number of discretization points used for each dimension) is at least quadratically improved compared to previous results for block-encoding the inverse matrix (see \Cref{rem:inv}). 
\end{itemize}

\subsection{Contents}
The paper is organized as follows. In \Cref{sec:prelim}, we specify the notation used in this paper and provide preliminary results needed in subsequent sections, such as quantum Fourier transform (QFT), the linear combination of unitaries (LCU), quantum signal processing (QSP) and quantum eigenvalue transformation (QET). In \Cref{sec:arithmD}, an algorithm is given for block encoding of generic symbols. For a separable symbol $a(\x,\z)= \alpha(\x)\beta(\z)$, a more efficient block encoding scheme is provided in \Cref{sec:separable}. Then a more explicit and practically feasible block encoding is developed in \Cref{sec:fully sep} for fully separable symbols of the form displayed in \eqref{eq:fullysepex}. Finally, \Cref{sec:app} presents the application of the proposed block encoding method to two types of widely used PDOs, including a variable coefficient second-order elliptic operator and the inverse of a constant coefficient elliptic operator. The paper is ended with a conclusion and discussion for future directions in \Cref{sec:con}. 

\subsection{Related works}

\subsubsection{Block encoding}
Most of the previous work \cite{childs2017quantum, berry2015simulating, low2017optimal} assumes that we have access to a matrix by querying two oracles that encode the locations and values of the non-zero elements of the objective matrix. Among them, \cite{gilyen2019quantum}*{Lemma 48} provides a general framework to explicitly construct the block encoding of sparse matrices if we are given these two oracles. Following this routine, \cite{camps2022explicit} constructs the block encoding of banded circulant matrices, extended binary tree matrices, and quantum walk operators. 

For general non-sparse matrices, it is clearly impossible to block-encode them in logarithmic time, and \cite{camps2022fable} proposes a near-optimal scheme for block encoding general unstructured matrices. Many methods are also proposed to implement the block encoding for full-rank dense matrices with certain structures, such as Toeplitz and Hankel systems \cite{mahasinghe2016efficient}, and linear group convolutions \cite{castelazo2022quantum} based on quantum Fourier transforms. The authors of \cite{nguyen2022quantum} introduce a new  method for kernel matrices with a hierarchical structure, which can be applied to non-uniform grids the Fourier transform cannot be used.

\subsubsection{Quantum PDE solvers} 
Along with the development of quantum linear system solvers \cite{harrow2009quantum, childs2017quantum, gilyen2019quantum, lin2020optimal, costa2022optimal}, many quantum PDE solvers are proposed to take advantage of exponential acceleration. Quantum counterparts of the finite element method (FEM) \cite{montanaro2016quantum} and the finite difference method (FDM) \cite{cao2013quantum, costa2019quantum} emerged for solving Poisson's equation and wave equation. In \cite{childs2021high}, adaptive finite difference and spectral methods are proposed to improve the dependence of the complexity on the error $\epsilon$ from $\mo(\poly(1 / \epsilon))$ to $\mo(\polylog (1 / \epsilon))$. It is worth noting that the process of block encoding the discretized differential operator is often not provided in these works, and constructing the block encoding for generic partial differential operators is highly non-trivial. 

\subsubsection{Numerical algorithms for PDOs}
There are also various classical numerical algorithms that compute PDOs efficiently. For example, \cite{demanet2011discrete} exploits the following expansion of symbols:
\[
a(\x,\z) \approx \sum_{j} \alpha_j(\x) \beta_j(\z).
\]
The paper presents efficient numerical approximations of $\beta_j(\z)$ with Chebyshev polynomials and hierarchical splines and further reduces the number of terms in the expansion by SVD or QR decomposition. However, a naive extension to high-dimensional PDOs leads to exponential overhead, as is the case for most classical methods. This is also one of the reasons why a quantum implementation of PDOs can be potentially useful.

\section{Preliminaries and notations}\label{sec:prelim}
\subsection{Notations}\label{subsec:note}

We adopt the commonly used notation for binary numbers: for an integer power of two $P=2^p$ and any $y\in\{0,\ldots,P-1\}$ if $y = y_0 + 2y_1 + \cdots+2^{p-1}y_{p-1} = (y_{p-1}y_{p-2}\cdots y_0.)$ in the binary system, the corresponding quantum state is $\ket{y} \equiv \ket{y_{p-1}\ldots y_0}$.

This extends to an $m$-tuple $\x = (x_1, \ldots, x_m)$ with $x_j\in\{0,\ldots,P-1\}$. The corresponding quantum state is given by $\ket{x_m}\cdots\ket{x_1}$, where $\ket{x_j} = \ket{x_{j,p-1} \cdots x_{j,0}}$ for each $j$. For a multivariate function $g:\{0,\ldots, P-1\}^m\rightarrow\R$, we denote by $D_{g}$ the diagonal multiplication operator on the Hilbert space $\mathbb{C}^{mp}$:
\begin{equation}\label{eq:diagg}
  \ket{x_{m}}\cdots\ket{x_{1}} \rightarrow g(x_{1}, \ldots, x_{m})\ket{x_{m}}\cdots\ket{x_{1}}.
\end{equation}
For a vector $\bm{v} = (v_1, v_2, \ldots, v_m)$, we denote by $\diag(\bm{v})$ the diagonal matrix with diagonal elements $(v_1,v_2, \ldots, v_m)$. 

The notation $|\bm{v}|$ for a $d$-dimensional vector $\bm{v}$ stands for the Euclidean norm $\sqrt{\sum_{j=1}^d|v_j|^2}$, where $v_j$ is the $j$-th coordinate of $\bm{v}$.

We also use the single qubit rotations $R_y(\theta)=e^{-i\frac{\theta}{2} Y} = \begin{bmatrix}\cos\frac{\theta}{2}&-\sin(\frac{\theta}{2})\\\sin\frac{\theta}{2}&\cos\frac{\theta}{2}\end{bmatrix}$ and $R_z(\theta)=e^{i\frac{\theta}{2}}e^{-i\frac{\theta}{2} Z} = \begin{bmatrix}1&\\&e^{i\theta}\end{bmatrix}$, where $Y$ and $Z$ are the Pauli matrices $Y=\begin{bmatrix}0&-i\\i&0\end{bmatrix}$ and $Z = \begin{bmatrix}1&\\&-1\end{bmatrix}$. The phase gate $S=\sqrt{Z}$ denotes the matrix $S = \begin{bmatrix}1&\\&i\end{bmatrix}$. To simplify the discussion, we assume access to all single qubit rotations, the Hadamard gate, the CNOT gate, the $2$-qubit SWAP gate, and the Toffoli gate when counting the number of elementary gates used. If one wants to use certain commonly used universal gate sets such as Hadamard and Toffoli, there will be some overhead linear in the number of gates involved and ploy-logarithm in the precision $\eps$, as bounded by the famous Solovay–Kitaev theorem \cite{kitaev1997quantum}. There are also many established results on decomposing commonly seen quantum gates with a certain universal gate set, such as \cite{rieffel2011quantum, da2022linear, vale2023decomposition}, to name a few. 

An $(m+n)$-qubit unitary operator $U$ is called a $(\gamma, m, \epsilon)$-block-encoding of an $n$-qubit operator $A$, if
$$
\left\|A-\gamma\left(\left\langle 0^m\right| \otimes I_n\right) U\left(\left|0^m\right\rangle \otimes I_n\right)\right\| \leq \epsilon,
$$
where $I_n$ denotes the $n$-qubit identity operator.
In the matrix form, a $(\gamma, m, \epsilon)$-block-encoding is a $2^{m+n}$ dimensional unitary matrix
$$
U=\left(\begin{array}{cc}
\widetilde{A} / \gamma & * \\
* & *
\end{array}\right)
$$
where $*$ can be any block matrices of the correct size and $\|\widetilde{A}-A\| \leq \epsilon$. In addition, when $A$ is a Hermitian matrix, it is possible to construct $U_A$ such that it is also Hermitian, in which case it is called a $(\gamma, m, \epsilon)$-Hermitian-block-encoding of $A$. The error $\epsilon$ is omitted in the notation of block encodings if $\epsilon=0$.

For an $n$-qubit system, the quantum Fourier transform (QFT) is an implementation of 
\begin{equation}\label{eq:QFT}
U_{\text{FT}}\ket{j} = \frac{1}{\sqrt{N}}\sum_{k=0}^{N-1} e^{2\pi i \frac{kj}{N}}\ket{k},
\end{equation}
where $N=2^n$, using a circuit $U_{\text{FT}}$ with $\mo(n^2)$ elementary gates and no ancilla qubit. The elementary gates involved include $2$-qubit swap gates and $2$-qubit controlled rotation gates. We refer the readers to \cite{coppersmith2002approximate, nielsen2002quantum} for more details on QFT. If only an approximation of $U_{\text{FT}}$ is needed, one can use approximated QFT \cite{nam2020approximate}, which has gate complexity $\mo(n\log(n/\eps))$, where $\eps$ is the spectral norm error of the approximation. 

\subsection{Linear combination of unitaries (LCU)}\label{subsec:LCU}
Given a few block-encoded matrices, a block encoding of a certain linear combination of them is often needed in practice. To this end, the linear combination of unitaries (LCU) technique has been developed (\cite{berry2015simulating, childs2017quantum, gilyen2019quantum}). For example, for two matrices $A$ and $B$, a block encoding of $A+B$ can be given by the circuit in \Cref{fig:LCU}, where $U_A$ and $U_B$ are block encodings of $A$ and $B$, respectively. 

\begin{figure}[ht]
\centerline{\Qcircuit @C=1em @R=1.5em {
\lstick{\ket{0}} & \gate{H} & \ctrlo{1} & \ctrl{1} & \gate{H}&\qw \\
\lstick{\ket{\psi}}& \qw & \gate{U_A} & \gate{U_B} & \qw& \qw \\
}
}    
\caption{LCU for two unitaries} \label{fig:LCU}                 
\end{figure}

For general linear combinations, we recall the following result from \cite{gilyen2019quantum} for general linear combinations. 
\begin{lemma}\label{lem:LCU}
For a vector $y\in\mathbb{C}^m$ with $\|y\|_1\leq\beta$, assume we have
\begin{enumerate}
    \item a pair of unitaries $(U_L, U_R)$ with shape $2^b\times2^b$ such that $U_L\ket{0^b} = \sum_{j=0}^{2^b-1}c_j\ket{j}$, $U_R\ket{0^b} = \sum_{j=0}^{2^b-1}d_j\ket{j}$, $\sum_{j=0}^{m-1}|\beta c_j^*d_j - y_j|<\epsilon_1$ and $\sum_{j=m}^{2^b-1}|c_j^*d_j|=0$, and
    \item a unitary $W=\sum_{j=0}^{m-1}\ket{j}\bra{j}\otimes U_j+\sum_{j=m}^{2^b-1}\ket{j}\bra{j}\otimes I_{a+s}$ where each $U_j$ is an $(\alpha, a, \epsilon_2)$-block-encoding of $A_j$ for $j=0, 1, \ldots, m-1$, 
\end{enumerate}
then $(U_L^\dag\otimes I_{a+s})W(U_R\otimes I_{a+s})$ is an $(\alpha\beta, a+b, \alpha\epsilon_1+\beta\epsilon_2)$-block-encoding of $\sum_{j=0}^{m-1}y_jA_j$, where $I_{a+s}$ denotes the identity operator with size $2^{a+s}\times2^{a+s}$. 
\end{lemma}

\subsection{Quantum eigenvalue transformation and quantum signal processing}\label{subsec:QET}
Given a Hermitian block encoding of Hermitian matrix $A$, one can construct a block encoding of $f(A)$ for a certain function $f$ using the qWeET) technique \cite{low2017optimal,gilyen2019quantum}. Let $f^e(\x)=\frac{1}{2}(f(\x)+f(-\x))$ and $f^o(\x)=\frac{1}{2}(f(\x)-f(-\x))$ be the even and odd part of $f(\x)$, respectively. The standard procedure consists of the following steps, where we assume that $f(\x)$ is properly scaled such that $\|f^e\|<1$, $\|f^o\|<1$, and $\|\cdot\|$ denotes the $L^\infty$ norm on $[-1,1]$.

\begin{enumerate}
    \item Approximate $f^e$ and $f^o$ with degree $\deg_{f^e}(\epsilon)$ even polynomial $\tilde{f}^e$ and degree $\deg_{f^o}(\epsilon)$ odd polynomial $\tilde{f}^o$, respectively, such that $\|\tilde{f}^e-f^e\|+\|\tilde{f}^o-f^o\|<\epsilon$ and $\|\tilde{f}^e\|\leq1$, $\|\tilde{f}^o\|\leq1$. 
    \item Find two sequences of phase factors $(\phi_0^e, \ldots, \phi_{\deg_{f^e}(\epsilon)}^e), (\phi_0^o, \ldots, \phi_{\deg_{f^o}(\epsilon)}^o)$ with each element in$[-\pi, \pi]$ using quantum signal processing (QSP) such that $\tilde{f}^e(\x)=\mathrm{Re}(p^e(\x))$, $\tilde{f}^o(\x)=\mathrm{Re}(p^o(\x))$, where $p^e$ and $p^o$ are complex polynomials with degree $\deg_{f^e}(\epsilon)$ and $\deg_{f^o}(\epsilon)$, respectively, given by
    \begin{equation}\label{eq:QSP}
    \begin{bmatrix}
        p(\x)&r(\x)\\r^*(\x)&p^*(\x)
    \end{bmatrix}=e^{i\phi_0Z}e^{i\arccos x X}e^{i\phi_1Z}e^{i\arccos x X}\cdots e^{i\phi_{\deg_f(\epsilon)-1}Z}e^{i\arccos x X} e^{i\phi_{\deg_f(\epsilon)}Z}.
    \end{equation}
    Here, the superscripts $e$ and $o$ are omitted for simplicity. The phase factors are then used in the QET circuit shown in \Cref{fig:QET} to construct block encodings $U_{f^e(A)}$ and $U_{f^o(A)}$, where the controlled rotation gate $\text{CR}_\phi$ is described in \Cref{fig:cr}
    \item Combine the block encodings $U_{f^e(A)}$ and $U_{f^o(A)}$ with the LCU circuit in \Cref{fig:LCU} to form the block encoding $U_{f(A)}$
\end{enumerate}

\begin{figure}[ht]
  \centering
   \subfigure[$\text{CR}_{\phi}$\label{fig:cr}.]{
    \Qcircuit @C=0.4em @R=1.5em {
  & \targ & \gate{e^{-i \phi Z}} & \targ    & \qw \\
& \ctrlo{-1} & \qw   & \ctrlo{-1} & \qw\\
} }\hspace{3em}
\subfigure[Quantum eigenvalue transformation.\label{fig:QET}]{
\Qcircuit @C=0.35em @R=1.2em {
\lstick{\ket{0}} &\gate{H}& \multigate{1}{\text{CR}_{{\phi}_{\t}}} & \qw & \multigate{1}{\text{CR}_{{\phi}_{\t-1}}} & \qw & \qw &&&\cdots&&&&\qw& \multigate{1}{\text{CR}_{{\phi}_0}}&\gate{H}&\qw\\
\lstick{\ket{0^m}} &\qw&\ghost{\text{CR}_{{\phi}_d}}& \multigate{1}{U_A} &  \ghost{\text{CR}_{{\phi}_{d-1}}}  & \multigate{1}{U_A} &\qw&&&\cdots &&&&\multigate{1}{U_A}&\ghost{\text{CR}_{{\phi}_0}}&\qw&\qw\\
\lstick{\ket{\psi}}&\qw& \qw& \ghost{U_A}& \qw& \ghost{U_A}& \qw&&&\cdots&&&&\ghost{U_A}&\qw&\qw&\qw\\
}}
\caption{(a): Circuit for the controlled rotation gate $\text{CR}_\phi$. (b): Circuit for quantum eigenvalue transformation. The circuit in (b) gives a block encoding for $U_{\mathrm{Re}(p(A))}$ based on the block encoding $U_A$ and phase factors $(\phi_0, \ldots, \phi_{\t})$, where $\t=\deg_f(\epsilon)$ and the polynomial $p(\x)$ and the phase factors satisfy \eqref{eq:QSP}.}
\label{fig:QETcr}
\end{figure}

There are several methods to find the phase factors $(\phi_0^e, \ldots, \phi_{\deg_f(\epsilon)}^e)$ and $(\phi_0^o, \ldots, \phi_{\deg_f(\epsilon)}^o)$ in $[-\pi, \pi]^{\deg_f(\epsilon)+1}$ in a stable and efficient way. For example, we refer to \cite{haah2019product,chao2020finding,dong2021efficient,ying2022stable,dong2022infinite} for more details. We summarize the procedure given above in \Cref{lem:QET} below, where we assume that $f$ is either even or odd for simplicity. 
\begin{lemma}\label{lem:QET}
For an even (resp. odd) function $f:\mathbb{R}\rightarrow \mathbb{R}$ and an $(\alpha,m)$-Hermitian-block-encoding of $A$ denoted as $U_A$, there is an $(\alpha C_f, m+1, \epsilon)$-block-encoding of $f(A)$ with gate complexity $\mo(\deg_f(\epsilon)(G_A+m))$ using the circuit shown in \Cref{fig:QET}, where $C_f\geq\max\{1, \|f\|\}$ is a scaling factor, $G_A$ is the gate complexity of $U_A$ and $\deg_f(\epsilon)$ is the smallest integer such that there exists an even (resp. odd) polynomial $\tilde{f}$ with a degree bounded by $\deg_f(\epsilon)$ satisfying $\|{f}-C_f\tilde{f}\|<\epsilon$ and $\|\tilde{f}\|\leq1$. The phase factors $(\phi_0, \ldots, \phi_{\deg_f(\epsilon)})$ in \Cref{fig:QET} are related with $\tilde{f}$ through \eqref{eq:QSP} and $\tilde{f}=\mathrm{Re}(p)$. 
\end{lemma}

\subsection{Discretization of pseudo-differential operators}\label{subsec:disc}
As mentioned in \Cref{sec:intro}, the PDO considered in this paper is defined for periodic functions on $\Omega = [0, 1]^d$:
\[
A f(\x) = \sum_{\z \in \mathbb{Z}^d} e^{2 \pi i \x \cdot \z} a(\x,\z) \hat{f}(\z).
\]
In most numerical treatments, the function $f$ is given on a discrete grid $X=\{\frac{\x}{P}\equiv(\frac{x_{1}}{P}, \ldots, \frac{x_{d}}{P}): x_j \in\{0, 1, \ldots, P-1\}\}$, where $P = 2^p$ is the number of discrete points used for each dimension. Notice that here we slightly abuse the notation by reusing $\x$ for the integer index of the grid points. Since the space variable takes value on the Cartesian grid $X$, the frequency domain is discretized correspondingly on $\{-\frac{P}{2}, \ldots, \frac{P}{2}-1\}^d$, which leads to the discretized PDO:
\begin{equation}\label{eq:pdoZ}
A f(\x) \equiv \sum_{\z \in \{-\frac{P}{2}, \ldots, \frac{P}{2}-1\}^d} e^{2 \pi i  \x\cdot \z/P} a(\frac{\x}{P},\z) \hat{f}(\z), \quad \x\in \Z,
\end{equation}
where $\Z=\{0,1, \ldots, P-1\}^d$. We adopt an abuse of notation and denote the discretized PDO by $A$ too.

Though the frequency variable $\z$ is discretized on $\{-\frac{P}{2}, \ldots, \frac{P}{2}-1\}^d$ in \eqref{eq:pdoZ}, by the convention of discrete Fourier transform (DFT) and fast Fourier transform (FFT), the frequency $(P/2, \ldots, P-1)$ is often identified with $(-P/2, \ldots, -1)$, respectively, since $P$ is a period for the frequency variable after DFT. In other words, the discretized PDO can be written as
\[
A f(\x) = \sum_{\z \in \Z} e^{2 \pi i  \x\cdot \z/P} \tilde{a}(\frac{\x}{P},\z) \hat{f}(\z), \quad \x\in \Z=\{0,1, \ldots, P-1\}^d,
\]
where 
\[
\tilde{a}(\x, \z) \equiv a\left(\x, \z-P\sum_{\z_j\geq P/2}\bm{e}_j\right),
\]
and $\bm{e}_j$ is the $j$-th standard basis vector in $\mathbb{C}^d$. As an example, when $d=1$, we have
\[
\tilde{a}(\x, \z) = 
\begin{cases}
a(\x, \z), \quad 0\leq\z<P/2,\\
a(\x, \z-P), \quad P/2\leq\z<P.
\end{cases}
\]
To simplify the notation and avoid repetitive use of $P$, we further define 
\begin{equation}
\a(\x, \z) \equiv \tilde{a}(\frac{\x}{P}, \z),
\end{equation}
and the discretized PDO becomes
\begin{equation}\label{eq:pdoZgen}
A f(\x) = \sum_{\z \in \Z} e^{2 \pi i  \x\cdot \z/P} \a({\x},\z) \hat{f}(\z), \quad \x\in \Z.
\end{equation}
It is clear that $\sup|\a|=\sup|a|$. In what follows, we also refer to $\a$ as the symbol of the PDO to be computed. 

\section{Block encoding for generic symbols}\label{sec:arithmD}
This section is concerned with the block encoding of the PDO \eqref{eq:pdoZ} (or rather \eqref{eq:pdoZgen}) with a generic symbol $a(\x,\z)$, without assuming any additional structure. In order to compute the PDO in \eqref{eq:pdoZgen}, a simple strategy is to first lift the state to the phase space $\Z\times\Z$. Then the multiplication of $\a(\x, \z)$ in \eqref{eq:pdoZgen} can be performed by diagonal matrix block encodings. Combining the QFT circuit and the block encoding of diagonal matrices, one can construct the entire circuit as illustrated in \Cref{fig:PDO1}. 

\begin{figure}[ht]
  \centerline{                                                                  
  \Qcircuit @C=1.5em @R=1.5em {
    \lstick{\ket{0^{pd}}}&\gate{H^{\otimes pd}}&\multigate{1}{U_{\text{ph}}^{\otimes d}}&\multigate{2}{U_\a}&\qw&&&\frac{1}{|Af|}\sum_{\x,\z}\a(\x, \z)e^{\frac{2\pi i \x\cdot\z}{P}}\hat{f}(\z)\ket{\x}\\
    \lstick{\frac{1}{|f|}\sum_{\x}f(\xp)\ket{\x}} & \gate{{U_{\text{FT}}^\dag}^{\otimes d}}&\ghost{U_{\text{ph}}^{\otimes d}}& \ghost{U_\a}&\gate{H^{\otimes pd}} &\qw&\meter \\
    \lstick{\ket{0^{b}}} & \qw&\qw& \ghost{U_\a}&\qw &\qw&\meter \\
}}    
\caption{Circuit that implements the PDO in \eqref{eq:pdoZgen} with a generic symbol. Here $b$ is the number of ancilla qubits needed for $U_\a$, $H$ is the Hadamard gate, $\frac{1}{|f|}\sum_{\x}f(\xp)\ket{\x}$ is the normalized input data , $U_{\text{FT}}$ is the QFT circuit, $U_{\text{ph}}$ and $U_\a$ are the circuits that perform the multiplication of $e^{2\pi i \x \cdot\z/P}$ and $\a(\x,\z)$ in \eqref{eq:pdoZgen}, and the desired output is obtained with normalizing factor $\frac{1}{|Af|}$ when getting $\ket{0^{pd+b}}$ for the $pd+b$ qubits on the bottom.}
\label{fig:PDO1}                                                         
\end{figure}

Now we explain the circuit displayed in \Cref{fig:PDO1} in more detail. First, we assume that the information of the function $f$ is prepared by a normalized vector $\frac{1}{|f|}\sum_{\x}f(\xp)\ket{\x}$, where $|f|$ is the normalization factor
\[|f|=\sqrt{\sum_{\x\in\Z}\left|f\left(\xp\right)\right|^2},\]
and $\Z=\{0,1, \ldots, P-1\}^d$. For functions $f$ with certain properties such as integrability, the state $\frac{1}{|f|}\sum_{\x}f(\xp)\ket{\x}$ can be constructed efficiently (see \cite{grover2002creating} for more details). For the rest of the paper, we assume the accessibility of the state $\frac{1}{|f|}\sum_{\x}f(\xp)\ket{\x}$ as an input. 

\textbf{Step 1. Apply QFT and lift the input state to the phase space.} 
We first obtain the representation of $f$ in the frequency domain by QFT. After applying the (inverse) QFT to the state $\frac{1}{|f|}\sum_{\x}f(\xp)\ket{\x}$ for $d$ times, we get 
\begin{equation}\label{eq:getxi}
\frac{1}{|f|}\frac{1}{\sqrt{P^d}}\sum_\z \sum_{\x} f(\xp)e^{-2\pi i \z \cdot \x/P}\ket{\z} = \frac{\sqrt{P^d}}{|f|}\sum_\z\hat{f}(\z)\ket{\z}. 
\end{equation}
Then the state $\frac{1}{|f|}\sum_{\x,\z}\hat{f}(\z)\ket{\x}\ket{\z}$ is obtained by applying the Hadamard gates $H^{\otimes pd}$ to the $\x$-register and putting both registers together. 

\textbf{Step 2. Multiply the phase $e^{2\pi i \x\cdot\z/P}$ with $U_{\text{ph}}$}.
A naive way of multiplying the phase $e^{2\pi i \x\cdot\z/P}$ is to use \Cref{prop:naivemult}, which involves many ancilla qubits and reduces the success probability. Here, we develop an efficient implementation for multiplication without involving any extra error or ancilla qubits in the following lemma. 

\begin{lemma}\label{lem:dot}
The $2p$-qubit circuit $U_{\text{ph}}$ displayed in \Cref{fig:lemma2} implements the unitary operator:
\begin{equation}
    \ket{x}\ket{\xi}\mapsto  e^{2\pi i x\xi/P}\ket{x}\ket{\xi}, \quad 0\le x,\xi<P,
\end{equation}
with $\mo(p^2)$ gate complexity precisely without ancilla qubits.
\begin{figure}[ht]
  \centerline{                                                                  
    \Qcircuit @C=.9em @R=0em @!R{  
      \lstick{\ket{x_0}} & \gate{R_p} & \gate{R_{p-1}} & \qw & \cdots & &  \gate{R_1}\    & \qw & \qw& \qw & \cdots &  & \qw& \qw & \cdots &  & \qw& \qw  \\      
      \lstick{\ket{x_1}} & \qw & \qw & \qw & \cdots & &  \qw\    & \gate{R_{p-1}} & \gate{R_{p-2}} & \qw & \cdots & &  \gate{R_1}    &\qw & \cdots & & \qw\ & \qw \\  
      \lstick{\vdots}&&&&\vdots &&&&&&\vdots &&&&\vdots  \\
      \lstick{\ket{x_{p-1}}} & \qw & \qw & \qw & \cdots & &  \qw\    & \qw & \qw & \qw & \cdots & &  \qw&\qw    &\cdots & & \gate{R_1} & \qw \\  
      \lstick{\ket{\xi_0}} & \ctrl{-4} & \qw & \qw & \cdots & &  \qw\    & \ctrl{-3} & \qw& \qw & \cdots &  & \qw& \qw & \cdots &  & \ctrl{-1} & \qw  \\      
      \lstick{\ket{\xi_1}} & \qw & \ctrl{-5} & \qw & \cdots & &  \qw\    & \qw & \ctrl{-4} & \qw & \cdots & &  \qw    &\qw & \cdots & & \qw\ & \qw \\  
      \lstick{\vdots}&&&&\vdots &&&&&&\vdots &&&&\vdots  \\
      \lstick{\ket{\xi_{p-2}}} & \qw & \qw & \qw & \cdots & &  \qw\    & \qw & \qw & \qw & \cdots & &  \ctrl{-6} &\qw    &\cdots & & \qw & \qw \\
      \lstick{\ket{\xi_{p-1}}} & \qw & \qw & \qw & \cdots & &  \ctrl{-8}\    & \qw & \qw & \qw & \cdots & &  \qw&\qw    &\cdots & & \qw & \qw \\
  }}    
\caption{Circuit for $U_{\text{ph}}$, the phase multiplication $\ket{x}\ket{\xi}\mapsto  e^{2\pi i x\xi/P}\ket{x}\ket{\xi}$. Here $R_{j} = \ket{0}\bra{0}+e^{2\pi i  \cdot2^{-j}}\ket{1}\bra{1}$ is a rotation operator.}
\label{fig:lemma2}                                                                       
\end{figure}
\end{lemma}

\begin{proof}
The idea of the construction is similar to the implementation of QFT, which is based on bit-wise controlled rotation. We first write the binary representation of integers 
\begin{equation*}
  x = (x_{p-1}\cdots x_0.),\quad \xi = (\xi_{p-1}\cdots\xi_0.),
\end{equation*}
and do the following calculation:
\begin{equation}\label{eq:Uph}
    \begin{aligned}
      &e^{2\pi i x\xi/P}\ket{x}\ket{\xi} 
      = \left(\prod_{j=0}^{p-1}\prod_{k=0}^{p-1}e^{2\pi i x_j\xi_k \cdot2^{j+k-p}}\right)\ket{x_{p-1}\cdots x_0}\ket{\xi_{p-1}\cdots \xi_0}\\
      =& \left(\prod_{0\le j+k<p}e^{2\pi i x_j\xi_k \cdot2^{j+k-p}}\right)\ket{x_{p-1}\cdots x_0}\ket{\xi_{p-1}\cdots \xi_0}\\
      =& \left(\left(\prod_{k=0}^{p-1}e^{2\pi i x_{0}\xi_{k} \cdot 2^{k-p}}\right)\ket{x_{0}}\right)\otimes\cdots\otimes\left(\left(\prod_{k=0}^{0}e^{2\pi i x_{p-1}\xi_{k} \cdot 2^{k-1}}\right)\ket{x_{p-1}}\right)\otimes\ket{\xi_{p-1}\cdots \xi_0},
    \end{aligned}
\end{equation}
where the second equality is true because $e^{2\pi i x_j\xi_k \cdot2^{j+k-p}}=1$ when $j+k\ge p$. The circuit corresponding to the unitary in \eqref{eq:Uph} can be implemented by a series of controlled rotations
\begin{equation*}
    \Qcircuit @C=1em @R=1em {
\lstick{\ket{x_j}} & \gate{R_{p-j-k}} & \rstick{e^{2\pi i x_j\xi_k \cdot2^{j+k-p}}\ket{x_j}} \qw \\
\lstick{\ket{\xi_k}} & \ctrl{-1} & \rstick{\ket{\xi_k}} \qw
}.
\end{equation*}
where $R_{j} =R_z(\pi/2^{j-1})= \ket{0}\bra{0}+e^{2\pi i  \cdot2^{-j}}\ket{1}\bra{1}$. Finally, the circuit shown in \Cref{fig:lemma2} is obtained after arranging the controlled rotations in the corresponding places. 
\end{proof}

Rewriting the state $\frac{1}{|f|}\sum_{\x,\z}e^{2\pi i \x\cdot\z/P}\hat{f}(\z)\ket{\x}\ket{\z}$ as 
\[
\frac{1}{|f|}\sum_{\x,\z}e^{2\pi i x_{d}\xi_{d}/P}\cdots e^{2\pi i x_{1}\xi_{1}/P}\hat{f}(\z)\ket{\x}\ket{\z},
\]
then the map from $\frac{1}{|f|}\sum_{\x,\z}\hat{f}(\z)\ket{\x}\ket{\z}$ to $\frac{1}{|f|}\sum_{\x,\z}e^{2\pi i \x\cdot\z/P}\hat{f}(\z)\ket{\x}\ket{\z}$ can be performed by applying \Cref{lem:dot} for $d$ times to the register pairs $(x_{d}, \xi_{d}), \ldots, (x_{1}, \xi_{1})$. The corresponding circuit is denoted as $U_\text{ph}^{\otimes d}$ and involves $\mo(p^2d)$ elementary gates with no ancilla qubits. After the multiplication of $U_{\text{ph}}^{\otimes d}$, one obtains the state $\frac{1}{|f|}\sum_{\x,\z}e^{2\pi i \x\cdot\z/P}\hat{f}(\z)\ket{\x}\ket{\z}$. 

\textbf{Step 3. Multiply the symbol $\a(\x,\z)$. }
The next component in \Cref{fig:PDO1} is the diagonal multiplication $U_\a$, which is designed to approximate the map
\begin{equation}\label{eq:Ua}
\frac{1}{|f|}\sum_{\x,\z}e^{2\pi i \x\cdot\z/P}\hat{f}(\z)\ket{\x}\ket{\z}\ket{0^b} \mapsto \frac{1}{C_a|f|}\sum_{\x,\z}\a(\x, \z)e^{2\pi i \x\cdot\z/P}\hat{f}(\z)\ket{\x}\ket{\z}\ket{0^b}+\perp,
\end{equation}
where $C_a>0$ is a constant that depends on $\a(\x,\z)$, $b$ is the number of ancilla qubits used for $U_\a$ and $\perp$ is an unnormalized state that is orthogonal to any state of the form $\ket{\x}\ket{\z}\ket{0^b}$.
As mentioned earlier, the idea is to treat $(\x, \z)$ as a $2d$-dimensional variable and utilize the result from arithmetic circuit construction. Leveraging the reversible computational model and the uncomputation technique, any classical arithmetic operation can be implemented by a quantum circuit efficiently. More specifically, one can construct a corresponding quantum circuit using $\mo(\polylog(\frac{1}{\epsilon}))$ ancilla qubits and $\mo(\polylog(\frac{1}{\epsilon}))$ gates, where $\epsilon$ is the precision one wants to achieve (Cf. \cite{nielsen2002quantum,rieffel2011quantum}). We state a general result for an efficient block encoding of diagonal matrices $D_g$, as defined in \eqref{eq:diagg}, which is summarized in \Cref{prop:diag1}. A similar idea has been used in \cite{grover2002creating} to create a given state, in \cite{harrow2009quantum} to construct the reciprocals of the eigenvalues and in \cite{tong2021fast} to implement the diagonal preconditioner.

\begin{prop}\label{prop:diag1}
Assume that $g:\R^m\rightarrow\R$ is a smooth arithmetic function with $\sup |g| <\infty$. Then there is an $(C, \mo(\polylog(\frac{1}{\epsilon})+\poly(mp)), \epsilon)$-block-encoding of $D_g$ with $\mo(\polylog(\frac{1}{\epsilon})+\poly(mp))$ gates, where $D_g$ is a diagonal operator on the Hilbert space $\mathbb{C}^{mp}$ defined in \eqref{eq:diagg}, and $C\geq\sup|g|$. 
\end{prop}
\begin{proof}
The circuit $U_g$ is constructed as follows. Let $t = \lceil \log_2(\frac{C\pi}{\epsilon})\rceil$, and let $\theta(x_1, \ldots, x_m)$ be a map that gives $\ket{\theta_{\text{sgn}}\theta_{t-1}\theta_{t-2}\ldots\theta_{0}}$, where $(.\theta_{t-1}\theta_{t-2}\ldots\theta_{0})$ is the closest $t$-bit fixed-point representation of $\frac{1}{\pi}\arcsin(|g(x_1, \ldots, x_m)|/C)$, and $\theta_{\text{sgn}}$ is assigned the value $0$ if $g\geq0$ and the value $1$ otherwise.  
For an arbitrary basis state 
$\ket{x_{m}}\cdots\ket{x_{1}}$, we consider the system with $t+2$ ancilla qubits $\ket{0}\ket{x_{m}}\cdots\ket{x_{1}}\ket{0^{t+1}}$. Here $\ket{x_j} = \ket{x_{j,p-1} \cdots x_{j,0}}$ for each $j$. Using the reversible computational model and uncomputation (\cite{nielsen2002quantum,rieffel2011quantum}), the classical circuit:
\[
\ket{0}\ket{x_{m}}\cdots\ket{x_{1}}\ket{0^{t+1}} \rightarrow \ket{0}\ket{x_{m}}\cdots\ket{x_{1}}\ket{\theta_{\text{sgn}}\theta_{t-1}\theta_{t-2}\ldots\theta_{0}}
\]
can be constructed with $\mo(\poly(t)+\poly(mp))$ gates and $\mo(\poly(mp))$ ancilla qubits. Then we apply the circuit in Figure \ref{fig:CR} on the $t+2$ ancilla qubits. The state obtained is:
\[(-1)^{\theta_{\text{sgn}}}(\cos(\pi(.\theta_{t-1}\ldots\theta_0))\ket{1}+\sin(\pi(.\theta_{t-1}\ldots\theta_0))\ket{0})\ket{x_{m}}\cdots\ket{x_{1}}\ket{\theta_{\text{sgn}}\theta_{t-1}\theta_{t-2}\ldots\theta_{0}},\]
which can then be mapped to 
\[(-1)^{\theta_{\text{sgn}}}(\cos(\pi(.\theta_{t-1}\ldots\theta_0))\ket{1}+\sin(\pi(.\theta_{t-1}\ldots\theta_0))\ket{0})\ket{x_{m}}\cdots\ket{x_{1}}\ket{0^{t+1}},\]
via uncomputation.

Since $|(.\theta_{t-1}\ldots\theta_0)-\frac{1}{\pi}\arcsin(|g(x_1, \ldots, x_m)|/C)|<\frac{\epsilon}{C\pi}$, we have
$|(-1)^{\theta_{\text{sgn}}}\sin(\pi(.\theta_{t-1}\ldots\theta_0))-\frac{1}{C}g(x_1, \ldots, x_m))|<\frac{\epsilon}{C}$, which means the desired state $\frac{1}{C}g(x_1, \ldots, x_m)\ket{x_{m}}\cdots\ket{x_{1}}$ is obtained with error at most $\frac{\epsilon}{C}$ upon measuring the ancilla qubits and getting $\ket{0^{t+2}}$. 
\end{proof}

\begin{figure}[ht]
  \centerline{                                                                  
    \Qcircuit @C=1.5em @R=1.5em {
      \lstick{\ket{0}}&\gate{Z}& \gate{R_{y}(\pi)} & \gate{R_{y}(\pi/2)} &\qw&\cdots&  & \gate{R_{y}(\pi/2^{t-1})}&\gate{X}&\qw\\
      \lstick{\ket{\theta_{\text{sgn}}}}& \ctrl{-1}& \qw & \qw&\qw&\qw&\qw&\qw&\qw&\qw\\
\lstick{\ket{\theta_{t-1}}}& \qw& \ctrl{-2} & \qw&\qw&\qw&\qw&\qw&\qw&\qw\\
\lstick{\ket{\theta_{t-2}}}& \qw& \qw& \ctrl{-3} &\qw& \qw&\qw&\qw&\qw&\qw\\
\lstick{\cdots}&\\
\lstick{\ket{\theta_{0}}}& \qw& \qw & \qw&\qw&\qw &\qw& \ctrl{-5} & \qw&\qw\\
}}    
\caption{The controlled rotation part of $U_g$.}
\label{fig:CR}         
\end{figure}

As a result of applying the block encoding given in \Cref{prop:diag1} with $m=2d$, $g = \a$ and $\epsilon$ replaced by $\epsilon/\sqrt{P^d}$ to the state $\frac{1}{|f|}\sum_{\x,\z}e^{2\pi i \x\cdot\z/P}\hat{f}(\z)\ket{\x}\ket{\z}\ket{0^b}$, one obtains the state
\[
\ket{\phi}\ket{0^b}+\perp,
\]
using $\mo(\polylog(\frac{1}{\epsilon})+\poly(pd))$ ancilla qubits and $\mo(\polylog(\frac{1}{\epsilon})+\poly(pd))$ gates where
\begin{equation}\label{eq:ineqphi}
\left\|\ket{\phi}-\frac{1}{C_a|f|}\sum_{\x,\z}\a(\x, \z)e^{2\pi i \x\cdot\z/P}\hat{f}(\z)\ket{\x}\ket{\z}\right\|<\frac{\epsilon}{\sqrt{P^d}C_a},
\end{equation}
and $C_a\geq\sup|a|=\sup|\a|$ is a constant. Here $\ket{\phi}$ denotes an unnormalized $2pd$-qubit state and $\perp$ is an unnormalized state orthogonal to all state with the form $\ket{\x}\ket{\z}\ket{0^b}$. 

\textbf{Step 4. Sum over the frequency variable.}
Finally, after applying the Hadamard gate $H^{\otimes pd}$ to the $\z$ registers, we obtain the state
\begin{equation}
\left((I_{pd}\otimes\ket{0^{pd}}) (I_{pd}\otimes\bra{0^{pd}})(I_{pd}\otimes H^{\otimes pd})\ket{\phi}\right)\ket{0^b}+\tilde{\perp},
\end{equation}
where $\tilde{\perp}$ is an unnormalized state that is orthogonal to all states of the form $\ket{\x}\ket{0^{pd+b}}$ and
\begin{equation}\label{eq:gen_final}
\begin{aligned}
    \bigg\|&\left((I_{pd}\otimes\ket{0^{pd}}) (I_{pd}\otimes\bra{0^{pd}})(I_{pd}\otimes H^{\otimes pd})\ket{\phi}\right)\ket{0^b}\\
    &-\frac{1}{\sqrt{P^d}|f|C_a}\sum_{\x,\z}\a(\x,\z)e^{2\pi i \x\cdot\z/P}\hat{f}(\z)\ket{\x}\ket{0^{pd}}\ket{0^b}\bigg\|<\frac{\epsilon}{\sqrt{P^d}C_a},
\end{aligned}
\end{equation} 
which can be seen from \eqref{eq:ineqphi}. Therefore, one obtains the desired state on the $\x$ registers upon measuring $\ket{0^{pd}}$ for the $\z$ registers and $\ket{0^{b}}$ for the ancilla qubits. Notice that there is an extra scaling factor $\frac{1}{\sqrt{P^d}}$ due to the application of Hadamard gates, so the success probability is $\mo(2^{-pd})$. The complete circuit we use is exactly the one shown in \Cref{fig:PDO1}.

The block encoding scheme of the PDO \eqref{eq:pdoZgen} constructed in this section can be summarized in the following theorem:
\begin{theorem}\label{thm:generica}
For a generic symbol $a(\x,\z)$, a block encoding of the corresponding discretized PDO defined by \eqref{eq:pdoZgen} can be $(2^{\frac{pd}{2}}C_a, \mo(\poly(pd)+\polylog(1/\epsilon)), \epsilon)$-block-encoded using the circuit displayed in \Cref{fig:PDO1} with gate complexity $\mo(\poly(pd)+\polylog(1/\epsilon))$, where $C_a\geq\sup|a|$ is a constant. 
\end{theorem}

\textbf{Challenge.}
Despite being applicable to generic symbols, one can observe from \eqref{eq:gen_final} that the success probability of the circuit in \Cref{fig:PDO1} can be low when $pd$ is large. In the following sections, we show that this challenge can be overcome when the symbol $a(\x,\z)$ has additional structures. 

\section{Efficient block encoding for separable symbols}\label{sec:separable}
As explained in \Cref{sec:arithmD}, the circuit designed as in \Cref{fig:PDO1} suffers from exponentially small success probability, despite being simple and applicable to generic PDOs. In this section, we are concerned with symbols with particular structures and an efficient block encoding of the corresponding PDOs with $\mo(1)$ success probability is constructed. 

\subsection{Separable symbols}
We first give the following definition for separable symbols. 
\begin{definition}\label{def:sep}
    A symbol $a(\x,\z)$ is \emph{separable} if $a(\x, \z) = \alpha(\x)\beta(\z)$.
\end{definition}

As explained in \Cref{subsec:disc}, we identify the frequency $(P/2, \ldots, P-1)$ with $(-P/2, \ldots, -1)$, respectively, since $P$ is a period for the frequency variable after DFT. We also define 
\begin{equation}\label{eq:sepnote}
\ba(\x) = \alpha(\xp), ~ \bb(\z) = \beta\left(\z-P\sum_{\z_j\geq P/2} \bm{e}_j\right),
\end{equation}
where $\bm{e}_j$ is the $j$-th standard basis vector in $\mathbb{C}^d$. 
With help of the notations \eqref{eq:sepnote}, the PDO \eqref{eq:pdoZ} becomes 
\begin{equation}\label{eq:pdoZsep}
A f(\x) = \sum_{\z \in \Z} e^{2 \pi i  \x\cdot \z/P} \ba(\x)\bb(\z) \hat{f}(\z), \quad \x\in \Z=\{0,1, \ldots, P-1\}^d,
\end{equation}
It is clear from the definition that $\ba$ is $P$-periodic since $\alpha$ is $1$-periodic, and we also have $\sup|\alpha|=\sup|\ba|$, $\sup|\beta|=\sup|\bb|$.

For the PDO \eqref{eq:pdoZsep}, we propose an efficient block encoding illustrated by the following circuit. 
\begin{figure}[ht]
\centerline{                                                                  
\Qcircuit @C=1.5em @R=1.5em {
    \lstick{\ket{0^{b_2}}} & \qw&\qw& \qw&\multigate{1}{U_\ba}\qw &\qw&\meter \\
    \lstick{\frac{1}{|f|}\sum_{\x}f(\xp)\ket{\x}} & \gate{{U_{\text{FT}}^\dag}^{\otimes d}}&\multigate{1}{U_\bb}& \gate{{U_{\text{FT}}}^{\otimes d}} &\ghost{U_\ba}&\qw&&&&\frac{1}{|Af|}\sum_{\x,\z}\ba(\x)\bb(\z)e^{2\pi i \x\cdot\z/P}\hat{f}(\z)\ket{\x} \\
    \lstick{\ket{0^{b_1}}} & \qw& \ghost{U_\bb}&\qw &\qw\qw&\qw&\meter \\
}
}    
\caption{Circuit for an efficient block encoding for separable PDOs \eqref{eq:pdoZsep}. Here $b_1$, $b_2$ are the number of ancilla qubits needed for $U_\bb$ and $U_\ba$, respectively, $\frac{1}{|f|}\sum_{\x}f(\xp)\ket{\x}$ is the normalized input data, $U_{\text{FT}}$ is the QFT circuit, $U_\bb$, and $U_\ba$ are the block encodings of $D_\bb$ and $D_\ba$, respectively, and the desired output is obtained with normalizing factor $\frac{1}{|Af|}$ when getting $\ket{0^{b_1+b_2}}$ for the $b_1+b_2$ ancilla qubits. $D_\bb$ and $D_\ba$ are diagonal matrices defined in \eqref{eq:diagg}.}\label{fig:PDO2}
\end{figure}

For the rest of this section, we explain the circuit in \Cref{fig:PDO2} with more details and show that it significantly improves the success probability compared with \Cref{fig:PDO1}. The circuit begins with a QFT step similar to that in \Cref{fig:PDO1}. 

\textbf{Step 1. Apply QFT and multiply the factor $\bb(\z)$.} 
With the same argument as \eqref{eq:getxi}, one obtains the state:
\[\frac{\sqrt{P^d}}{|f|}\sum_\z\hat{f}(\z)\ket{\z}\ket{0^{b_1+b_2}},\]
after applying the QFT gates ${U_{\text{FT}}^\dag}^{\otimes d}$. Now that the input is represented on the frequency domain, one can naturally implement the multiplication of the factor $\bb(\z)$ since it only depends on the frequency variable $\z$. The corresponding block $U_\bb$ can be constructed using \Cref{prop:diag1} with $m=d$, $g=\bb$ and $\epsilon$ replaced by $\frac{\epsilon}{2C_\alpha}$, where the constant $C_\alpha\geq\sup|\alpha|$. Then one obtains the state 
\begin{equation}\label{eq:state1}
(\ket{\phi_1}\ket{0^{b_1}}+\perp_1)\ket{0^{b_2}},
\end{equation}
where $\perp_1$ is an unnormalized state orthogonal to all states of the form $\ket{\z}\ket{0^{b_1}}$ and $\ket{\phi_1}$ satisfies
\begin{equation}\label{eq:b}
\left\|\ket{\phi_1}-\frac{\sqrt{P^d}}{C_\beta|f|}\sum_\z\bb(\z)\hat{f}(\z)\ket{\z}\right\|<\frac{\epsilon}{2C_\alpha C_\beta},
\end{equation}
using $\mo(\poly(pd)+\polylog(1/\epsilon))$ gates and $b_1=\mo(\poly(pd)+\polylog(1/\epsilon))$ ancilla qubits. Here $C_\beta\geq\sup|\beta|$ is a constant.

\textbf{Step 2. Apply QFT and multiply the factor $\ba(\x)$. }
In order to multiply the factor $\ba(\x)$ in the symbol, we apply QFT and convert the state to the space domain. Since $\sqrt{P^d}{U_{\text{FT}}}^{\otimes d}\sum_\z\bb(\z)\hat{f}(\z)\ket{\z} = \sum_{\x,\z}\bb(\z)e^{2\pi i \x\cdot\z/P}\hat{f}(\z)\ket{\x}$, we have 
\begin{equation}\label{eq:phib}
\begin{aligned}
&\left\|{U_{\text{FT}}}^{\otimes d}\ket{\phi_1}-\frac{1}{C_\beta|f|}\sum_{\x,\z}\bb(\z)e^{2\pi i \x\cdot\z/P}\hat{f}(\z)\ket{\x}\right\|\\
&=\left\|{U_{\text{FT}}}^{\otimes d}\ket{\phi_1}-{U_{\text{FT}}}^{\otimes d}\frac{\sqrt{P^d}}{C_\beta|f|}\sum_\z\bb(\z)\hat{f}(\z)\ket{\z}\right\|\\
&=\left\|\ket{\phi_1}-\frac{\sqrt{P^d}}{C_\beta|f|}\sum_\z\bb(\z)\hat{f}(\z)\ket{\z}\right\|<\frac{\epsilon}{2C_\alpha C_\beta},
\end{aligned}
\end{equation}
where we have used \eqref{eq:b} in the last line. 
By using \Cref{prop:diag1} again with $m=d$, $g=\ba$ and $\epsilon$ replaced by $\frac{\epsilon}{2C_\beta}$, the state $\ket{\phi_1}\ket{0^{b_2}}$ is mapped to $\ket{\phi_2}\ket{0^{b_2}}+\perp_2$, where $\perp_2$ is an unnormalized state orthogonal to all state of the form $\ket{\x}\ket{0^{b_2}}$ and $\ket{\phi_2}$ satisfies $\|\ket{\phi_2}-\frac{1}{C_\alpha}{D}_\ba{U_{\text{FT}}}^{\otimes d}\ket{\phi_1}\|<\frac{\epsilon}{2C_\alpha C_\beta}$. In this step, $\mo(\poly(pd)+\polylog(1/\epsilon))$ gates and $b_2=\mo(\poly(pd)+\polylog(1/\epsilon))$ ancilla qubits are used. The image of $\perp_1$ is still orthogonal to all states of the form $\ket{\z}\ket{0^{b_1}}$ since the $b_1$ ancilla qubits used in the previous step are unchanged. Therefore, the final state is 
\begin{equation}\label{eq:state2}
\ket{\phi_2}\ket{0^{b_1+b_2}}+\perp,
\end{equation}
where $\perp$ is an unnormalized state orthogonal to all states of the form $\ket{\x}\ket{0^{b_1+b_2}}$ and $\ket{\phi_2}$ satisfies
\begin{equation}\label{eq:a}
\begin{aligned}
&\left\|C_\alpha C_\beta\ket{\phi_2}-\frac{1}{|f|}\sum_{\x,\z}\ba(\x)\bb(\z)e^{\frac{2\pi i \x\cdot\z}{P}}\hat{f}(\z)\ket{\x}\right\|\\
\leq&C_\alpha C_\beta\left\|\ket{\phi_2}-\frac{1}{C_\alpha}{D}_\ba{U_{\text{FT}}}^{\otimes d}\ket{\phi_1}\right\|+C_\beta\left\|D_\ba{U_{\text{FT}}}^{\otimes d}\ket{\phi_1}-\frac{1}{C_\beta|f|}\sum_{\x,\z}\ba(\x)\bb(\z)e^{\frac{2\pi i \x\cdot\z}{P}}\hat{f}(\z)\ket{\x}\right\|\\
<& \frac{\epsilon}{2} + {C_\beta}\left\|D_\ba{U_{\text{FT}}}^{\otimes d}\ket{\phi_1}-D_\ba\frac{1}{C_\beta|f|}\sum_{\x,\z}\bb(\z)e^{2\pi i \x\cdot\z/P}\hat{f}(\z)\ket{\x}\right\|\\
\leq& \frac{\epsilon}{2} + C_\beta\sup |\ba|\left\|{U_{\text{FT}}}^{\otimes d}\ket{\phi_1}-\frac{1}{C_\beta|f|}\sum_{\x,\z}\bb(\z)e^{2\pi i \x\cdot\z/P}\hat{f}(\z)\ket{\x}\right\|<\frac{\epsilon}{2} + \frac{\epsilon}{2} =\epsilon,
\end{aligned}
\end{equation}
where we have used the inequality \eqref{eq:phib} and the fact that $C_\alpha \geq \sup|\alpha| = \sup|\ba|$ in the last line. 
By checking the definition of block encoding and adding up the gates and ancilla qubits used, we obtain the following theorem.

\begin{theorem}\label{thm:sep}
If $a(\x,\z)=\alpha(\x)\beta(\z)$ is a separable symbol as defined in \Cref{def:sep}, then the discretized PDO \eqref{eq:pdoZsep} can be $({C_\alpha C_\beta}, \mo(\poly(pd)+\polylog(1/\epsilon)), \epsilon)$-block-encoded using the circuit displayed in \Cref{fig:PDO2} with gate complexity $\mo(\poly(pd)+\polylog(1/\epsilon))$, where $C_\alpha, C_\beta > 0$ are constants such that $C_\alpha\geq\sup|\alpha|$ and $C_\beta\geq\sup|\beta|$. 
\end{theorem}
\begin{remark}
In contrast to the result in \Cref{thm:generica}, one can observe that the exponential factor $2^{\frac{pd}{2}}$ is removed, and thus the success probability for the circuit in \Cref{fig:PDO2} is improved exponentially compared to the one in \Cref{fig:PDO1}. 
\end{remark}

\subsection{Linear combination of separable terms}
With the block encoding for PDOs with separable symbols ready, the PDO for a linear combination of separable terms, i.e.,
\[
a(\x,\z) = \sum_{j=0}^{m-1} y_ja_j(\x,\z) =\sum_{j=0}^{m-1} y_j\alpha_j(\x)\beta_j(\z),
\]
can also be block-encoded, thanks to LCU (see \Cref{subsec:LCU}). More precisely, we have the following corollary. 
\begin{cor}\label{cor:sep}
    For a linear combination of separable symbols $a(\x,\z) = \sum_{j=0}^{m-1} y_ja_j(\x,\z) =\sum_{j=0}^{m-1} y_j\alpha_j(\x)\beta_j(\z)$, where $\sup|\alpha_j|\leq1$, $\sup|\beta_j|\leq1$ and $y\in\mathbb{C}^m$ with $\|y\|_1\leq\delta$, assume that $(U_L, U_R)$ is a pair of unitaries described in \Cref{lem:LCU} with $\epsilon_1=\epsilon$, and $W=\sum_{j=0}^{m-1}\ket{j}\bra{j}\otimes U_j+\sum_{j=m}^{2^b-1}\ket{j}\bra{j}\otimes I_{a+s}$, where each $U_j$ is a $(1,a, \epsilon)$-block-encoding of the discretized PDO $A_j$ associated with symbol $a_j(\x,\z)$ (see \eqref{eq:pdoZsep}) constructed in \Cref{thm:sep} with $a=\mo(\poly(pd)+\polylog(1/\epsilon))$. Then $(U_L^\dag\otimes I_{a+s})W(U_R\otimes I_{a+s})$ is a $(\delta, a+b, (1+\delta)\epsilon)$-block-encoding of $\sum_{j=0}^{m-1}y_jA_j$, where $I_{a+s}$ denotes the identity operator with size $2^{a+s}\times2^{a+s}$. The gate complexity of the corresponding circuit is $\mo(m(\poly(pd)+\polylog(1/\epsilon)))$. 
\end{cor}

\section{Efficient block encoding for fully separable symbols with explicit circuits}\label{sec:fully sep}

For separable symbols, the circuit presented in \Cref{fig:PDO2} significantly increases the success probability compared to the one in \Cref{fig:PDO1}. However, this relies on circuits for arithmetic functions (see \Cref{prop:diag1}), which can still be challenging to construct in practice. In this section, we develop a more explicit circuit with the help of QSP and QET (see \Cref{subsec:QET}). 

\subsection{Dimension-wise fully separable symbols}
To begin with, we consider the fully separable symbols defined as follows. 
\begin{definition}\label{def:fullysep}
  A symbol $a(\x,\z)$ is called \emph{fully separable} if $a(\x,\z) = \alpha(\x)\beta(\z)$ with $\alpha(\x) = \alpha_1(x_1)\cdots\alpha_d(x_d)$ and $\beta(\z) = \beta_1(\xi_1)\cdots\beta_d(\xi_d)$ where each function in the set $\{\alpha_k\}_{k=1}^d\cup\{\beta_k\}_{k=1}^d$ is a real even function, a real odd function or an exponential function of the form $f(y)=\exp(i\theta y)$ for some real parameter $\theta$.
\end{definition}

Similar with \eqref{eq:sepnote}, we introduce the following notations:
\begin{equation}\label{eq:fullynote}
  \ba_k(x_k) = \alpha_k(\xkp), ~
  \bb_k(\xi_k) = \begin{cases}
    \beta_k(\xi_k),   & 0\leq\xi_k<P/2\\
    \beta_k(\xi_k-P), & P/2\leq\xi_k<P
  \end{cases},\quad k= 1, 2, \ldots, d.
\end{equation}
Then the discretized PDO \eqref{eq:pdoZ} becomes 
\begin{equation}\label{eq:pdoZfully}
A f(\x) = \sum_{\z \in \Z} e^{2 \pi i  \x\cdot \z/P} \left(\prod_{k=1}^d\ba_k(x_k)\right)\left(\prod_{k=1}^d\bb_k(\xi_k)\right) \hat{f}(\z), \quad \x\in \Z=\{0,1, \ldots, P-1\}^d.
\end{equation}
In order to block-encode the PDO \eqref{eq:pdoZfully}, we adopt the following circuit, as shown in \Cref{fig:PDO3}, which is similar with the one in \Cref{fig:PDO2}: 
\begin{figure}[ht]
\centerline{                                                                  
\Qcircuit @C=1.5em @R=1.5em {
    \lstick{\ket{0^{2d}}} & \qw&\qw& \qw&\multigate{1}{\bigotimes_{k=1}^d U_{\ba_k}}\qw &\qw&\meter \\
    \lstick{\frac{1}{|f|}\sum_{\x}f(\xp)\ket{\x}} & \gate{{U_{\text{FT}}^\dag}^{\otimes d}}&\multigate{1}{\bigotimes_{k=1}^d U_{\bb_k}}& \gate{{U_{\text{FT}}}^{\otimes d}} &\ghost{\bigotimes_{k=1}^d U_{\ba_k}}&\qw&&\frac{1}{|Af|}\sum_{\x}Af(\xp)\ket{\x} \\
    \lstick{\ket{0^{2d}}} & \qw& \ghost{\bigotimes_{k=1}^d U_{\bb_k}}&\qw &\qw\qw&\qw&\meter \\
}
}    
\caption{Circuit for an explicit and efficient block encoding for fully separable PDOs \eqref{eq:pdoZfully}. 
Here $\frac{1}{|f|}\sum_{\x}f(\xp)\ket{\x}$ is the normalized input data, $U_{\text{FT}}$ is the QFT circuit, $U_{\bb_k}$ and $U_{\ba_k}$ are the block encodings of $D_{\bb_k}$ and $D_{\ba_k}$, respectively, and the desired output is obtained with normalizing factor $\frac{1}{|Af|}$ when getting $\ket{0^{b_1+b_2}}$ for the $b_1+b_2$ ancilla qubits. Here $D_{\bb_k}$ and $D_{\ba_k}$ are diagonal matrices defined in \eqref{eq:diagg}.}\label{fig:PDO3}
\end{figure}

Exploiting the fully separable structure of the symbol, one can construct explicit circuits for the diagonal multiplications shown in \Cref{fig:PDO3} by leveraging QSP and QET (see \Cref{subsec:QET}). To this end, we first introduce a lemma that gives Hermitian block encodings for two diagonal multiplication prototypes that allow us to build the QET circuit afterward. More specifically, by combining a series of single-qubit rotations, one can construct a diagonal matrix with diagonal elements $\{\exp(ij\theta)\}_{j=0}^{P-1}$. Then one can build a diagonal matrix with diagonal elements $\{\sin(ij\theta)\}_{j=0}^{P-1}$ using a simple LCU circuit, from which a diagonal matrix with diagonal elements $\{g(ij\theta)\}_{j=0}^{P-1}$ can be built with QET for some smooth function $g$. Since the grid points of the frequency variables are taken as $\{-\frac{P}{2}, -\frac{P}{2}+1, \ldots, \frac{P}{2}-1\}$, one also needs the corresponding diagonal matrices where the index $j$ takes values in $\{-\frac{P}{2}, -\frac{P}{2}+1, \ldots, \frac{P}{2}-1\}$ rather than from $0$ to $P$. We denote the corresponding matrices by subscript $-$ as opposed to $+$ if the index $j$ goes from $0$ to $P$. During the preparation of this paper, we notice that a similar result is built in \cite{mcardle2022quantum} independently. 

\begin{lemma}\label{lem:diagv}
For a fixed $P=2^p$, let $\bm{v}_-$, $\bm{v}_+$ be the vectors:
$$\bm{v}_- = \left(0,1,\ldots,\frac{P}{2}-1,-\frac{P}{2},-\frac{P}{2}+1,\ldots,-1\right),$$ 
and
$$\bm{v}_+ = \left(0,1,\ldots,\frac{P}{2}-1,\frac{P}{2},\frac{P}{2}+1,\ldots,P-1\right),$$ 
respectively, and $D_-$, $D_+$ be the diagonal matrices $\diag\!(\bm{v}_-)$, $\diag\!(\bm{v}_+)$, respectively. Then there is an $(1, 1)$-Hermitian-block-encoding of $\sin(\theta D_\sigma)$ with gate complexity $2p+5$, where $\sigma\in\{-,+\}$ and $\theta>0$ is a parameter. 
\end{lemma}
\begin{proof}
For an arbitrary parameter $\theta>0$, we first construct matrices $R_- = \exp( i\theta D_-)$ and $R_+ = \exp( i\theta D_+)$ with quantum circuits. Using the binary representation of $\xi$, we get
\[
R_-\ket{\xi_{p-1}\cdots\xi_0} = e^{i((-\xi_{p-1})\xi_{p-2}\cdots\xi_0.)\theta}\ket{\xi_{p-1}\cdots\xi_0},
\]
and
\[
R_+\ket{\xi_{p-1}\cdots\xi_0} = e^{i(\xi_{p-1}\cdots\xi_0.)\theta}\ket{\xi_{p-1}\cdots\xi_0}.
\]
Then 
\[
\begin{aligned}
R_+\ket{\xi_{p-1}\cdots\xi_0} &= e^{i(\xi_{p-1}\cdots\xi_0.)\theta}\ket{\xi_{p-1}\cdots\xi_0}=e^{i\sum_j\xi_j2^j\theta}\ket{\xi_{p-1}\cdots\xi_0}\\
& = \prod_je^{i\xi_j2^j\theta}\bigotimes_j\ket{\xi_j} = \bigotimes_je^{i\xi_j2^j\theta}\ket{\xi_j} =  \bigotimes_jR_z(2^j\theta)\ket{\xi_j},
\end{aligned}
\]
and
\[
R_-\ket{\xi_{p-1}\cdots\xi_0} = e^{-i\xi_{p-1}2^p\theta}R_+\ket{\xi_{p-1}\cdots\xi_0} = \bigotimes_jR_z((-1)^{\delta_{j, p-1}}2^j\theta)\ket{\xi_j},
\]
so $R_+ = \bigotimes_jR_z(2^j\theta)$ and $R_- = \bigotimes_jR_z((-1)^{\delta_{j, p-1}}2^j\theta)$, where $R_z$ is the single qubit rotation defined in \Cref{subsec:note}. Let $U_\sigma$ be the circuit displayed in \Cref{fig:D}, where $\sigma\in\{-,+\}$, then $U_\sigma$ is a $(1, 1)$-Hermitian-block-encoding of $\sin(\theta D_\sigma)$. 
\begin{figure}[ht]
\centerline{                                                                  
\Qcircuit @C=1em @R=1.5em {
\lstick{\ket{0}}&\gate{S} & \gate{H} & \ctrlo{1} & \ctrl{1} & \gate{H}& \gate{X}& \gate{S} &\meter  \\
\lstick{\ket{\psi}}&\qw & \qw & \gate{R_\sigma} & \gate{R_\sigma^\dag} & \qw& \qw& \qw &\qw &&\sin(\theta D_\sigma)\ket{\psi} \\
}
}    
\caption{Hermitian block encoding of $\sin(\theta D_\sigma)$.} \label{fig:D}     
\end{figure}
In fact, the matrix corresponding to $U_\sigma$ is
\[
\begin{aligned}
U_\sigma &= \begin{bmatrix}I&\\&iI\end{bmatrix}\begin{bmatrix}&I\\ I&\end{bmatrix}\frac{1}{\sqrt{2}}\begin{bmatrix}I&I\\ I&-I\end{bmatrix}\begin{bmatrix}R_\sigma&\\&R_\sigma^\dag\end{bmatrix}\frac{1}{\sqrt{2}}\begin{bmatrix}I&I\\ I&-I\end{bmatrix}\begin{bmatrix}I&\\&iI\end{bmatrix}\\
&= \begin{bmatrix}I&\\&iI\end{bmatrix}\begin{bmatrix}&I\\ I&\end{bmatrix}\frac{1}{2}\begin{bmatrix}R_\sigma+R_\sigma^\dag&R_\sigma-R_\sigma^\dag\\R_\sigma-R_\sigma^\dag&R_\sigma+R_\sigma^\dag\end{bmatrix}\begin{bmatrix}I&\\&iI\end{bmatrix}\\
&= \frac{1}{2}\begin{bmatrix}R_\sigma-R_\sigma^\dag&i(R_\sigma+R_\sigma^\dag)\\i(R_\sigma+R_\sigma^\dag)&-(R_\sigma-R_\sigma^\dag)\end{bmatrix} = i\begin{bmatrix}\sin(\theta D_\sigma)&\cos(\theta D_\sigma)\\\cos(\theta D_\sigma)&-\sin(\theta D_\sigma)\end{bmatrix},\\
\end{aligned}
\]
and this is a Hermitian matrix with the first diagonal block being $\sin(\theta D_\sigma)$ after ignoring the global phase factor $i$, since $\begin{bmatrix}\sin(\theta D_\sigma)&\cos(\theta D_\sigma)\\\cos(\theta D_\sigma)&-\sin(\theta D_\sigma)\end{bmatrix}$ is Hermitian. It can be seen from \Cref{fig:D} that the number of elementary gates used is $2p+5$. 
\end{proof}

Now, we aim to construct the block encoding of diagonal matrices $D_{\ba_j}=\ba_j(D_+)$ and $D_{\bb_j}=\bb_j(D_+)=\beta_j(D_-)$, namely $U_{\ba_k}$ and $U_{\bb_k}$ in \Cref{fig:PDO3}. For the case where $\alpha_j(x_j)=\exp(i\theta x_j)$ or $\beta_j(\xi_j)=\exp(i\theta \xi_j)$, the matrices $R_- = \exp( i\theta D_-)$ and $R_+ = \exp( i\theta D_+)$ constructed in \Cref{lem:diagv} are exactly the diagonal matrices $D_{\bb_j}=\beta_j(D_-)$ and $D_{\ba_j}=\ba_j(D_+)$, respectively. Therefore, we devote the rest of this section to the case where $\beta_j$ and $\alpha_j$ are even or odd real functions. Thanks to the block encodings $U_+$ and $U_-$ introduced in \Cref{lem:diagv}, what remains to do is to find polynomial approximations of $\ba$ and $\beta$ so as to complete the QET procedure described in \Cref{subsec:QET}. Specifically, we restrict the parameter $\theta$ to be $0<\theta<\frac{\pi}{2P}$ in order to recover $\theta D_\sigma$ from $\sin(\theta D_\sigma)$ with the arcsin function. Going through the QET procedure, one obtains the following result. 

\begin{prop}\label{prop:QET}
Let $U_-$ and $U_+$ be the $(1, 1)$-Hermitian-block-encodings of $\sin(\theta D_-)$ and $\sin(\theta D_+)$ constructed in \Cref{lem:diagv} for $0<\theta<\frac{\pi}{2P}$ with $P=2^p$. Assume that $g$ is an even (resp. odd) continuous real function on $[-P, P]$ and that $\deg_g(\epsilon)$ is the smallest positive integer such that there exists an even (resp. odd) polynomial $\tilde{g}$ with the degree bounded by $\deg_g(\epsilon)$ satisfying 
\begin{equation}\label{eq:qspasp}
\sup_{-\sin(P\theta)\leq x\leq \sin(P\theta)}\left|C_g\tilde{g}(\x)-g\left(\frac{1}{\theta}\arcsin x\right)\right|<\epsilon, ~ \|\tilde{g}\|\leq 1,
\end{equation}
where $\|\cdot\|$ is the $L^\infty$ norm on $[-1,1]$ and $C_g$ is a constant such that $C_g\geq\sup|g|$. 
Then there is a $(C_g, 2, \epsilon)$-block-encoding for both $g(D_-)$ and $g(D_+)$ with $\mo(p\deg_g(\epsilon))$ gates, where $D_-$ and $D_+$ are defined in \Cref{lem:diagv}.
\end{prop}

\begin{proof}
The circuit in \Cref{fig:QET} with $U_A$ replaced by $U_\sigma$ gives a $(C_g, 2)$-block-encoding for $\tilde{g}(\sin(\theta D_\sigma))$, and since 
\[
\sup_{-\sin(P\theta)\leq x\leq \sin(P\theta)}\left|C_g\tilde{g}(\x)-g\left(\frac{1}{\theta}\arcsin x\right)\right|<\epsilon,
\]
we have 
\[
\left\|C_g\tilde{g} (\sin(\theta D_\sigma))-g\left(\frac{1}{\theta}\arcsin (\sin(\theta D_\sigma))\right)\right\|<\epsilon,
\]
where the operator $2$-norm is used.
Thus the circuit in \Cref{fig:QET} with $U_A$ replaced by $U_\sigma$ gives a $(C_g, 2, \epsilon)$-block-encoding for
\[
g\left(\frac{1}{\theta}\arcsin (\sin(\theta D_\sigma))\right) = g\left(\frac{1}{\theta}\cdot\theta D_\sigma\right) = g(D_\sigma),
\]
where $\sigma\in\{-,+\}$ and we have used the fact that $0<\theta<\frac{\pi}{2P}$ in the first equality. Since $\mo(p)$ gates are used in $U_\sigma$, the gate complexity of the circuit described above is $\mo(p\deg_g(\epsilon))$, which closes the proof. 

\end{proof}

The gate complexity of the circuit built with QET in \Cref{prop:QET} depends on the smoothness of $g$. For instance, we have the following corollary. 
\begin{cor}\label{cor:analytic}
Assume that $g$ is an even (resp. odd) differentiable real function on $[-\frac{\pi}{2\theta},\frac{\pi}{2\theta}]$. Let $\dg{g}{\epsilon}{\theta}$ be the smallest integer such that there exists a polynomial $u$ with degree $\dg{g}{\epsilon}{\theta}$ satisfying $\sup_{|y|<\frac{\pi}{2\theta}}|u(y)-g(y)|<\frac{\epsilon}{3}$, then the gate complexity of the circuit used in \Cref{prop:QET} is 
$$\mo\left(p\log\left(\frac{C_g'}{\theta\epsilon}\right)\dg{g}{\epsilon}{\theta}\right),$$
where $C_g' = \sup_{|y|<\frac{\pi}{2\theta}}|g'(y)|$. In particular, if $g$ is a polynomial, the gate complexity reduces to $\mo\left(p\log\left(\frac{C_g'}{\theta\epsilon}\right)\deg g\right)$.
\end{cor}
\begin{proof}
Without loss of generality, assume $1>\frac{\epsilon}{2C_g}$, otherwise we can just let $\tilde{g} = 0$ \eqref{eq:qspasp}. Since $\frac{1}{\theta}\arcsin(x)$ is an analytic function whose power series centered at $x=0$ has convergence radius $1>\sin(P\theta)$, there is a truncation $v$ of the Taylor series of $\frac{1}{\theta}\arcsin(x)$ with degree $\mo(\log\left(\frac{C_g'}{\theta\epsilon}\right))$ such that $\sup_{|x|< \sin(P\theta)}|v(x)-\frac{1}{\theta}\arcsin(x)|<\frac{\epsilon}{3C_g'}$. Now since the coefficients of the Taylor series of $\frac{1}{\theta}\arcsin(x)$ at $x=0$ are all non-negative, it holds that $v([-1,1])\subset [-\frac{\pi}{2\theta},\frac{\pi}{2\theta}]$. Let $\bar{g}(x)=\frac{1}{C_g}u(v(x))$ and $\tilde{g} = (1-\frac{\epsilon}{3C_g})\bar{g}$, then $\tilde{g}$ has degree $\mo(\log\left(\frac{C_g'}{\theta\epsilon}\right)\dg{g}{\epsilon}{\theta})$, and
\[
\begin{aligned}
    \abs{C_g\bar{g}(x)-g\left(\frac{1}{\theta}\arcsin(x)\right)}
    &\leq \abs{u(v(x))-g(v(x))}+\abs{g(v(x))-g\left(\frac{1}{\theta}\arcsin(x)\right)}\\
    &< \frac{\epsilon}{3} + \sup_{|y|<\frac{\pi}{2\theta}}|g'(y)|\abs{v(x)-\frac{1}{\theta}\arcsin(x)} 
    < \frac{\epsilon}{3} + C_g'\cdot\frac{\epsilon}{3C_g'}=\frac{2\epsilon}{3},\\
\end{aligned}
\]
for $x\in [-\sin(P\theta), \sin(P\theta)]$. In addition, since $v$ maps $[-1,1]$ into $[-\frac{\pi}{2\theta},\frac{\pi}{2\theta}]$, we have 
\[
\begin{aligned}
|\bar{g}(x)|&\leq\sup_{|y|<\frac{\pi}{2\theta}}\abs{\frac{1}{C_g}u(y)}
\leq \sup_{|y|<\frac{\pi}{2\theta}}\left|\frac{1}{C_g}u(y)-\frac{1}{C_g}g(y)\right|+\sup_{|y|<\frac{\pi}{2\theta}}\left|\frac{1}{C_g}g(y)\right|  
\le\frac{\epsilon}{3C_g}+1,
\end{aligned}
\]
and therefore $\abs{\tilde{g}(x)} = (1-\frac{\epsilon}{3C_g})\abs{\bar{g}(x)}<1$ for $x\in[-1,1]$. In addition, we have
\[
\begin{aligned}
    \abs{C_g\tilde{g}(x)-g\left(\frac{1}{\theta}\arcsin(x)\right)}
    &= \abs{(1-\frac{\epsilon}{3C_g})\left(C_g\bar{g}(x)-g\left(\frac{1}{\theta}\arcsin(x)\right)\right)+\frac{\epsilon}{3C_g}g\left(\frac{1}{\theta}\arcsin(x)\right)}\\
    &\leq (1-\frac{\epsilon}{3C_g})\frac{2\epsilon}{3}+\frac{\epsilon}{3C_g}C_g
    < \frac{2\epsilon}{3} + \frac{\epsilon}{3}=\epsilon,\\
\end{aligned}
\]
for $x\in [-\sin(P\theta), \sin(P\theta)]$. According to \Cref{prop:QET}, we have $\deg_g(\epsilon)=\mo(\log\left(\frac{C_g'}{\theta\epsilon}\right)\dg{g}{\epsilon}{\theta})$ and the gate complexity of the circuit used in \Cref{prop:QET} is $\mo\left(p\log\left(\frac{C_g'}{\theta\epsilon}\right)\dg{g}{\epsilon}{\theta}\right)$, where the factor $p$ comes from preparing $\sin(\theta D_{\sigma})$ as mentioned in \cref{lem:diagv}. In the case that $g$ is a polynomial, we can simply let $u=g$ and thus $\dg{g}{\epsilon}{\theta} \leq \deg g$.
\end{proof}

For the final step, we first introduce the notation
\begin{equation}\label{eq:de}
\de{\epsilon} \equiv \sum_{k=1}^d\left[\deg_{\ba_k}(\epsilon)+\deg_{\beta_k}(\epsilon)\right],
\end{equation}
where $\deg_{\ba_k}(\epsilon)$ and $\deg_{\beta_k}(\epsilon)$ are defined in \Cref{prop:QET}. Denote by $U_{\ba_k}$ and $U_{\bb_k}$ the block encodings of $\ba_k(D_+)$ and $\bb_k(D_+)=\beta_k(D_-)$ obtained in \Cref{prop:QET} with $\epsilon$ replaced by $\epsilon/2dC$, respectively, where 
\begin{equation}\label{eq:const}
    C = \widetilde{C}_\alpha \widetilde{C}_\beta, ~ \widetilde{C}_\alpha=\prod_{k=1}^d C_{\ba_k}, ~ \widetilde{C}_\beta=\prod_{k=1}^dC_{\bb_k},
\end{equation}
and $C\geq\prod_{k=1}^d(\sup|\ba_k|\sup|\bb_k|)$ since $C_{\ba_k}\geq\sup|\ba_k|$ and $C_{\bb_k}\geq\sup|\bb_k|$. 
Now we are ready to prove the following theorem that relies on the block encodings $\bigotimes_{k=1}^d U_{\ba_k}$ and $\bigotimes_{k=1}^d U_{\bb_k}$ in \Cref{fig:PDO3}. 

\begin{theorem}\label{thm:fullysep}
If $a(\x, \z)=\alpha(\x)\beta(\z)= \alpha_1(x_1)\cdots\alpha_d(x_d)\beta_1(\xi_1)\cdots\beta_d(\xi_d)$ is a fully separable symbol as defined in \Cref{def:fullysep}, then the corresponding PDO defined by \eqref{eq:pdoZfully} can be $(C, \mo(d), \epsilon)$-block-encoded with gate complexity $\mo(dp\de{\frac{\epsilon}{2dC}}+dp^2)$ using the circuit displayed in \Cref{fig:PDO3}, where $C > 0$ is a constant defined in \eqref{eq:const}, and $\de{\frac{\epsilon}{2dC}}$ is defined in \eqref{eq:de}.  
\end{theorem}
\begin{proof}
Since each $U_{\ba_k}$ is a $(C_{\ba_k}, 2, \epsilon/2dC)$-block-encoding for $\ba_k(D_+)$ according to \Cref{prop:QET}, $\bigotimes_{k=1}^d U_{\ba_k}$ is a $(\prod_{k=1}^d C_{\ba_k}, 2d, \epsilon/2\widetilde{C}_\beta)$-block-encoding for 
\[\bigotimes_{k=1}^d\ba_k(D_+)=\bigotimes_{k=1}^d D_{\ba_k}=D_\ba.\]
Similarly, $\bigotimes_{k=1}^d U_{\bb_k}$ is a $(\prod_{k=1}^d C_{\bb_k}, 2d, \epsilon/2\widetilde{C}_\alpha)$-block-encoding for 
\[\bigotimes_{k=1}^d\bb_k(D_+)=\bigotimes_{k=1}^d D_{\bb_k}=D_\bb.\]
Hence by the same argument as the proof of \Cref{thm:sep} (especially \eqref{eq:state1}, \eqref{eq:b}, \eqref{eq:phib}, \eqref{eq:state2} and \eqref{eq:a}), the circuit in \Cref{fig:PDO3} gives a $(C, 4d, \epsilon)$ block encoding of the PDO \eqref{eq:pdoZfully} with gate complexity $\mo(p\de{\frac{\epsilon}{2dC}}+p^2d)$, where the $\mo(p^2d)$ term comes from the QFT part of the circuit.
\end{proof}
\begin{remark}
    The approximate QFT (see \cite{nam2020approximate} for example) can be used to replace the QFT blocks in \Cref{fig:PDO3}. With similar arguments as in the proof of \Cref{thm:fullysep}, one can show that the gate complexity can be reduced to $\mo(dp\de{\frac{\epsilon}{2dC}})$. 
\end{remark}

\subsection{Linear combination of fully separable terms}
Similar to \Cref{cor:sep}, we can block-encode the PDO for a linear combination of fully separable terms, i.e.,
\[
a(\x,\z) = \sum_{j=0}^{m-1} y_ja_j(\x,\z) =\sum_{j=0}^{m-1} y_j\alpha_{j1}(x_1)\cdots\alpha_{jd}(x_d)\beta_{j1}(\xi_1)\cdots\beta_{jd}(\xi_d),
\]
with LCU (see \Cref{subsec:LCU}) and \Cref{thm:fullysep}, which is stated in the following corollary. 
\begin{cor}\label{cor:fullysep}
    For a linear combination of fully separable symbols $a(\x,\z) = \sum_{j=0}^{m-1} y_ja_j(\x,\z) =\sum_{j=0}^{m-1} y_j\alpha_{j1}(x_1)\cdots\alpha_{jd}(x_d)\beta_{j1}(\xi_1)\cdots\beta_{jd}(\xi_d)$, where $\underset{[0,1]}{\sup}~|\alpha_{jk}|\leq1$, $\underset{[-P/2,P/2]}{\sup}|\beta_{jk}|\leq1$ and $y\in\mathbb{C}^m$ with $\|y\|_1\leq\delta$, assume that $(U_L, U_R)$ is a pair of unitaries described in \Cref{lem:LCU} with $\epsilon_1=\epsilon$, and $W=\sum_{j=0}^{m-1}\ket{j}\bra{j}\otimes U_j+\sum_{j=m}^{2^b-1}\ket{j}\bra{j}\otimes I_{a+s}$, where each $U_j$ is a $(1,a, \epsilon)$-block-encoding of the discretized PDO $A_j$ associated with symbol $a_j(\x,\z)$ (see \eqref{eq:pdoZfully}) constructed in \Cref{thm:fullysep} with $a=\mo(d)$. Then $(U_L^\dag\otimes I_{a+s})W(U_R\otimes I_{a+s})$ is an $(\delta, a+b, (1+\delta)\epsilon)$-block-encoding of $\sum_{j=0}^{m-1}y_jA_j$, where $I_{a+s}$ denotes the identity operator with size $2^{a+s}\times2^{a+s}$. The gate complexity of the corresponding circuit is $\mo(dp\sum_{j=0}^{m-1}\de{\frac{\epsilon}{2d}}+dp^2m)$.
\end{cor}

\section{Applications}\label{sec:app}
In this section, we provide worked examples for particular symbols using the circuit shown in  \Cref{fig:PDO3} and provide complexity analysis, beginning with a variable coefficient second-order elliptic operator. 

\subsection{Second-order elliptic operator with variable coefficients}\label{subsec:varellip}
Recall that the elliptic operator introduced in \eqref{eq:ellip} is of the following form: 
\begin{equation}
(Au)(\x) = u(\x)-\nabla\cdot (\omega(\x)\nabla u(\x)).\label{eq:elliptic}
\end{equation}
In this section, we assume that $\omega(\x)>0$ has a low-rank Fourier expansion 
\begin{equation}\label{eq:lowrank}
\omega(\x) = \sum_{j=1}^r c_j \exp\left(2\pi i \q_j \cdot \x \right),\quad \q_j\in\mathbb{Z}^d.
\end{equation}
Many commonly seen functions have low-rank expansions or approximations. For instance, $\omega(\x) = 2+\sin(2\pi\sum_{l=1}^d x_l) > 0$ can be written in the rank-3 form 
\[\omega(\x) = 2 + \frac{ i}{2}(\exp(-2\pi i(x_1+\cdots+x_d))-\exp(2\pi i(x_1+\cdots+x_d))).\]

Plugging the form \eqref{eq:lowrank} of $\omega$ into \eqref{eq:ellipsym}, one obtains the symbol associated with the PDO above
\[
a(\x,\z) = 1+\sum_{j=1}^r\sum_{l=1}^d(4\pi^2Pq_{jl}c_j)e^{2\pi  i \q_j\cdot \x}\frac{\xi_l}{P}+\sum_{j=1}^r\sum_{l=1}^d(4\pi^2P^2c_j)e^{2\pi  i \q_j\cdot \x}\frac{\xi_l^2}{P^2},
\]
where $P=2^p$ is the number of discrete points used for each dimension (see \Cref{subsec:disc}). Notice that the terms $e^{2\pi  i \q_j\cdot \x}\frac{\xi_l}{P}$ and $e^{2\pi  i \q_j\cdot \x}\frac{\xi_l^2}{P^2}$ above are fully separable by \Cref{def:fullysep}, thus by \Cref{cor:fullysep}, we know that the corresponding PDO can be block-encoded. As explained in \Cref{sec:fully sep}, the multiplication of $e^{2\pi  i \q_j\cdot \x}=\prod_{l=1}^de^{2\pi  i q_{jl} x_l}$ can be implemented directly using $R_-$ and $R_+$ constructed in \Cref{lem:diagv}. 
Consequently, the number of gates needed for multiplying each $e^{2\pi iq_{jl} x_l}$ factor without error is $\mo(p)$, and no ancilla qubits are used. Since $\frac{\xi_l}{P}$ and $\frac{\xi_l^2}{P^2}$ are polynomials, by \Cref{cor:analytic}, the multiplication of each $\frac{\xi_l}{P}$ and $\frac{\xi_l^2}{P^2}$ factor can be implemented with $\mo\left(p\log\left(\frac{1}{\epsilon}\right)\right)$ gates to $\mo(\epsilon)$ precision, and $\mo(d)$ ancilla qubits are used. Going through the proof of \Cref{thm:fullysep}, one can see that the PDO associated with $e^{2\pi  i \q_j\cdot \x}\frac{\xi_l}{P}$ and $e^{2\pi  i \q_j\cdot \x}\frac{\xi_l^2}{P^2}$ can be $(1, \mo(d),\epsilon)$-block-encoded with gate complexity $\mo(p\log(\frac{1}{\epsilon})+p^2+dp)$, where the three terms account for implementing the polynomials of $\xi_l$, the QFT of the $l$-th component, and the multiplication of $e^{2\pi  i \q_j\cdot \x}$, respectively. Finally, going through the LCU step as in \Cref{cor:fullysep} with $\mo(dr)$ terms, one obtains a $(\gamma, \mo(d+\log(dr)), (1+\gamma)\epsilon)$-block-encoding of the PDO \eqref{eq:ellip} with total gate complexity $\mo(dr(p\log\left(\frac{1}{\epsilon}\right)+p^2+dp))=\mo\left(dpr(\log\frac{1}{\epsilon}+d+p)\right)$, where $\gamma=1+4\pi^2(P\sum_{j=1}^r|c_j|\norm{\q_j}_1+P^2d\sum_{j=1}^r |c_j|)$. This result is summarized in the following theorem, where we have used $\mo(d+\log(dr))=\mo(d+\log(r))$.




\begin{theorem}\label{thm:elliptic}
    For the elliptic operator \eqref{eq:ellip} with variable coefficient, where $\omega(\x)$ has a low-rank expansion \eqref{eq:lowrank}, there exists a $(\gamma, \mo(d+\log(r)), (1+\gamma)\epsilon)$-block-encoding for the corresponding discretized PDO defined in \eqref{eq:pdoZfully} with gate complexity 
    \[\mo\left(dr\log P\left(\log\frac{P}{\epsilon}+d\right)\right),\]
    where $\gamma=1+4\pi^2(P\sum_{j=1}^r|c_j|\norm{q_j}_1+P^2d\sum_{j=1}^r |c_j|)$.
\end{theorem}
Similar to previous sections and by a slight abuse of notation, we denote the discretization of the operator defined in \eqref{eq:elliptic} also by $A$. Now we can use the QLSA in \cite{costa2022optimal} to get the following corollary:
\begin{cor}
    Let $(A, b)$ be the discretization of the operator and the right-hand side of \eqref{eq:elliptic}, respectively, there is a quantum algorithm finding the normalized state $\ket{A^{-1}b}=\frac{A^{-1}b}{\|A^{-1}b\|}$ within error $\epsilon$ with gate complexity 
    \[
    \mo\left(\gamma dr\log\frac{1}{\epsilon}\log P\left(\log\frac{\gamma P}{\epsilon}+d\right)\right).
    \]
\end{cor}
\begin{proof}
It is clear that $\norm{A} \le \gamma$ and $\norm{A^{-1}} \le 1$. According to \cref{thm:elliptic}, one can construct $U_A$, the $(\gamma, \mo(d+\log(r)), \epsilon/\gamma)$-block-encoding of $A$, with complexity
\[\mo\left(dr\log P\left(\log\frac{\gamma P}{\epsilon}+d\right)\right).\] 
In other words, $U_A$ is a $(1, \mo(d+\log(r)), 0)$-block-encoding of some matrix $\tilde{A}/\gamma$ such that $\|A-\tilde{A}\|<\epsilon/\gamma$. Therefore, we have $\norm{\tilde{A}} = \mo(\gamma)$, $\norm{\tilde{A}^{-1}} = \mo(1)$, and thus $\kappa(\tilde{A}) = \mo(\gamma)$. The main theorem of \cite{costa2022optimal} gives a quantum algorithm that can output a state $\ket{y}$ that is $\mo(\epsilon)$ close to $\ket{(\tilde{A}/\gamma)^{-1}b} = \ket{\tilde{A}^{-1}b} = \frac{\tilde{A}^{-1}b}{\|\tilde{A}^{-1}b\|}$, using
\[
\mo\left(\kappa(\tilde{A})\log{\frac{1}{\epsilon}}\right) = \mo\left(\gamma\log{\frac{1}{\epsilon}}\right)
\]
queries of $U_A$. Finally, we have the estimation
\begin{equation*}
\begin{aligned}
    \norm{\ket{A^{-1}b}-\ket{\tilde{A}^{-1}b}} &= \left\|\frac{A^{-1}b}{\norm{A^{-1}b}}-\frac{\tilde{A}^{-1}b}{\norm{\tilde{A}^{-1}b}}\right\|\\
    &\le \frac{\norm{(A^{-1}-\tilde{A}^{-1})b}}{\norm{A^{-1}b}} + \norm{\tilde{A}^{-1}b}\left|\frac{1}{\norm{A^{-1}b}}-\frac{1}{\norm{\tilde{A}^{-1}b}}\right|\\
    &\le \frac{\norm{A^{-1}}\norm{\tilde{A}^{-1}}\norm{A-\tilde{A}}\norm{b}}{\norm{A}^{-1}\norm{b}} + \frac{\norm{A^{-1}}\norm{\tilde{A}^{-1}}\norm{A-\tilde{A}}\norm{b}}{\norm{A}^{-1}\norm{b}}\\
    &= \mo\left(\gamma\cdot\frac{\epsilon}{\gamma}\right) = \mo(\epsilon).
\end{aligned}
\end{equation*}
So $\ket{y}$ is also an $\mo(\epsilon)$ approximation of $\ket{A^{-1}b}$ and the overall gate complexity is
\[
    \mo\left(\gamma dr\log\frac{1}{\epsilon}\log P\left(\log\frac{\gamma P}{\epsilon}+d\right)\right).
\]
\end{proof}

\subsection{Application for constant coefficient elliptic operator}\label{subsec:appinverse}
In this part, we investigate the multiplier operators (see \eqref{eq:multiplier}), i.e., the PDOs with symbols of the form $a(\x,\z) = \beta(\z)$. In particular, we showcase how to directly block-encode the inverse of multiplier operators based on the results obtained in the previous sections so that one can solve a discretized system of a PDE without invoking QLSAs (quantum linear system algorithms). 

In practice, it is quite often the case that $\beta(\z)$ is radially symmetric. For example, many operators related to Laplacian have radial symmetric symbols since the symbol of $\Delta$ is $-4\pi^2|\z|^2$. For the rest of this section, we focus on the radially symmetric symbol $\beta(\z)=\varphi(|\z|)$. The main idea of dealing with radially symmetric $\beta$ is to consider an approximation in the following form
\begin{equation}\label{eq:fullysepidea}
    \beta(\z)=\varphi(|\z|)\approx \sum_{m=1}^M w_m e^{-a_m |\z|^2}=\sum_{m=1}^M w_m\prod_{i=1}^d e^{-a_m \xi_i^2}.
\end{equation}
Notice that the right-hand side is in the fully separable form by \Cref{def:fullysep}. Thus the results from \Cref{sec:fully sep} can be used. In fact, the authors of \cite{beylkin2005approximation} developed an efficient algorithm to find a low-rank approximation of a single-variable function $f(y)$ by exponential sums
\begin{equation}
    f(y)\approx\sum_{m=1}^M w_m e^{-a_m y^2},\label{eq:exp approx}
\end{equation} 
for a large range of even functions $f(y)$, especially those whose amplitude decrease as $y\to\infty$. 

Before diving into concrete examples, we first introduce a near-optimal polynomial approximation to the exponential functions in the following lemma, which is a direct corollary of \cite{sachdeva2014faster}*{Theorem 4.1}. 

\begin{lemma}\label{lem:approx}
  For every $a, b>0$, and $0<\delta \leq 1$, there exists an even polynomial $r(y)$ satisfying
  \[
  \sup _{y \in[0, b]}\left|e^{-ay^2}-r(y)\right| \leq \delta,\quad \sup _{y \in[0, b]}\left|r(y)\right| \leq 1,
  \]
  and has degree $\mo\left(\sqrt{\max \{ab^2, \log 1 / \delta\} \cdot \log 1 / \delta}\right)$.
\end{lemma}
With \Cref{lem:approx} and established results in \Cref{sec:fully sep}, we are ready to give the following block encoding result for the diagonal multiplication $D_\beta$ associated with $\beta(\z)$ (see \eqref{eq:diagg}). 

\begin{theorem}\label{thm:ex2}
If $\beta(\z)$ can be approximated in the following sense
\[
\left|\beta(\z) - \sum_{m=1}^M w_m e^{-a_m |\z|^2}\right| \leqslant \epsilon, \quad \text { for all } \z\in\left[-\frac{P}{2},\frac{P}{2}\right]^d,
\]
with $w_m, a_m\ge 0$ and $R := \max_m a_m$, $W := \sum_{m=1}^M|w_m|$, then we can implement a $(\gamma,q,\epsilon)$-block-encoding for the PDO associated with $\beta$ (defined in \eqref{eq:pdoZ}) with
\[
\gamma = \mo(W),\quad q = \mo\left(d+\log M\right),
\]
and gate complexity 
\[
\mo\left(dM\log P\log\left(\frac{dWRP}{\epsilon}\right)\sqrt{\max\left\{RP^2,\log\frac{Wd}{\epsilon}\right\}\log\frac{Wd}{\epsilon}}\right).
\]
\end{theorem}
\begin{proof}
For each exponential term $g_m(z) = e^{-a_mz^2}$, we know 
\[
\dg{g_m'}{\epsilon}{\theta}= \mo\left(\sqrt{\max\left\{\frac{R}{\theta^2},\log\frac{1}{\epsilon}\right\}\log\frac{1}{\epsilon}}\right) 
\]
according to \Cref{lem:approx}, where $\dg{g_m'}{\epsilon}{\theta}$ is defined in \Cref{cor:analytic}. Moreover, a simple calculation shows that
\begin{equation}
C_{g_m'} \leq \sup |g_m'(z)| = \sup \abs{2a_m z e^{-a_mz^2}} < \sqrt{a_m}\le \sqrt{R}.
\end{equation}
Therefore, if we let $\theta = \frac{\pi}{3P}$, then the complexity of implementing a $(1,1,\frac{\epsilon}{2Wd})$-block-encoding of matrix $\exp(-a_m (D_-)^2)$ is
$$\mo\left(\log P\log\left(\frac{dWRP}{\epsilon}\right)\sqrt{\max\left\{RP^2,\log\frac{Wd}{\epsilon}\right\}\log\frac{Wd}{\epsilon}}\right),$$
according to \cref{cor:analytic}.

After implementing $\exp(-a_m (D_-)^2)$ for each $\xi_k$ register ($k=1,\ldots,d$), we get a $(1,d,\frac{\epsilon}{2W})$-block-encoding of $e^{-a_m|\z|^2}$. Since there are $M$ such terms, the total gate complexity after conducting LCU becomes
\begin{equation}\label{eq:pdeconst}
\mo\left(dM\log P\log\left(\frac{dWRP}{\epsilon}\right)\sqrt{\max\left\{RP^2,\log\frac{Wd}{\epsilon}\right\}\log\frac{Wd}{\epsilon}}\right).
\end{equation}
Since only one QFT and one iQFT are needed with complexity $\mo(d\log^2 P)$, the dominating term in the above complexity formula remains unchanged. The total error of this block encoding is $\sum_{m=1}^M |w_m|\frac{\epsilon}{2W} < \epsilon$ as desired. Since $\mo(d)$ ancillae are used when encoding $e^{-a_m|\z|^2}$ and $\mo(\log M)$ ancillae are used for LCU, the total number of ancilla is $\mo(d+\log M)$.
\end{proof}

With \Cref{thm:ex2}, we are ready to work on a concrete example. Consider the following $d$-dimensional elliptic equation with periodic boundary conditions: 
\begin{equation}\label{eq:PDEex}
-\frac{1}{4\pi^2}\Delta u(\x) + u(\x) = b(\x),\quad x\in[0,1]^d.
\end{equation}
As explained by \eqref{eq:fullysepidea} and \eqref{eq:exp approx}, the idea is to expand $\varphi(y)=\frac{1}{1+y^2}$ as the sum of a series of exponential functions. 
To this end, we introduce a
result for exponential approximations in the form of \eqref{eq:fullysepidea} and \eqref{eq:exp approx}. 
\begin{lemma}\label{lem:exp1}
  For any $0<\delta \leqslant 1$ and $0<\epsilon \leqslant \half$, there exist positive numbers $p_m$ and $v_m$ such that
  \begin{equation}
    \left|r^{-1}-\sum_{m=1}^M v_m e^{-p_m r}\right| \leqslant r^{-1} \epsilon, \quad \text { for all } \delta \leqslant r \leqslant 1, \label{eq:exp2}
  \end{equation}
  with
  \begin{equation}
    M=\mo\left(\log \epsilon^{-1}\left(\log \epsilon^{-1}+\log \delta^{-1}\right)\right),\label{eq:exp_m}
  \end{equation}
  and
  \begin{equation}
    \max_m p_m = \mo\left(\delta^{-1}\log \epsilon^{-1}\left(\log \epsilon^{-1}+\log \delta^{-1}\right)\right).\label{eq:exp_p}
  \end{equation}
\end{lemma}

The proof of \eqref{eq:exp_m} can be found in Theorem A.1 of \cite{beylkin2005approximation}, and \eqref{eq:exp_p} is also implied in the proof there. We summarize the construction given in \cite{beylkin2005approximation} in passing. Denote $f_r(t) = e^{-r e^t+t}$, then $r^{-1} = \int_{-\infty}^{\infty}f_r(t)\mathrm{d} t$ and $f_r(t)$ decays rapidly when $|t|\rightarrow\infty$. Therefore one can approximate this integral using the trapezoidal rule on a finite interval $[a,b]$ with step size $h$, which is
\begin{equation}
  h\left(\sum_{k=1}^{K-1} f_r(a+k h)+\frac{f_r(a)+f_r(b)}{2}\right),\label{eq:trap}
\end{equation}
where $K = (b-a)/h$. Since each term of \eqref{eq:trap} is of the form $ve^{-pr}$, the approximate integral formula above actually provides an approximation of the form \eqref{eq:exp2}. Finally, by choosing $a = -\log\frac{4}{\epsilon}$, $b = \log \left(4\log\frac{4}{\epsilon} \delta^{-1} \log \left(2\log\frac{4}{\epsilon}(\delta \epsilon)^{-1}\right)\right)$ and $h\le \pi/(2\log\frac{4}{\epsilon}+1)$, one can show that the condition \eqref{eq:exp2} is satisfied, and from $M = K+1 = (b-a)/h+1$ one can check that \eqref{eq:exp_m} and \eqref{eq:exp_p} hold. 

Now, with the substitutions $\delta = (\frac{dP^2}{4}+1)^{-1}$ and $r = (1+y^2)/(\frac{dP^2}{4}+1)$ in \Cref{lem:exp1} we get the following approximation:
\begin{equation}\label{eq:exp3}
  \left|\frac{1}{1+y^2}-\sum_{m=1}^M w_m e^{-a_m y^2}\right| \leqslant  \epsilon, \quad \text { for all } -\frac{\sqrt{d}P}{2} \leqslant y \leqslant \frac{\sqrt{d}P}{2},
\end{equation}
where $a_m = p_m/(\frac{dP^2}{4}+1)$, $w_m = e^{-a_m}v_m/(\frac{dP^2}{4}+1)$, and
\begin{equation}
\begin{aligned}
M &= \mo\left(\log \epsilon^{-1}\left(\log \epsilon^{-1}+\log (dP)\right)\right),\\
R &= \max_m a_m = \mo\left(\log \epsilon^{-1}\left(\log \epsilon^{-1}+\log (dP)\right)\right),\\
W &= \sum_{m=1}^M |w_m| = \mo(1).
\end{aligned}
\end{equation}
Here the estimation of $W$ is deduced from plugging $y=0$ into \eqref{eq:exp3} and the positivity of $w_m$. Finally, after plugging the estimations of $M$, $R$ and $W$ into \Cref{thm:ex2}, we obtain the following result:
\begin{cor}\label{thm:invmult}
For the pseudo-differential operator associated with the symbol $\beta(\z)=1+|\z|^2$ (see \eqref{eq:PDEex}), there is a $(\gamma,q,\epsilon)$-block-encoding for its inverse with $\gamma = \mo(1)$, $q = \mo\left(d+\log\log\frac{dP}{\epsilon}\right)$, and gate complexity $\mo\left(dP\left(\log\frac{dP}{\epsilon}\right)^{2.5}\left(\log\frac{1}{\epsilon}\right)^{1.5}\left(\log\frac{d}{\epsilon}\right)^{0.5}\log P\right)$.
\end{cor}
\begin{remark}\label{rem:inv}
      If one uses the uniform grid to discretize the operator $(-\frac{1}{4\pi^2}\Delta + 1)$ in equation \eqref{eq:PDEex}, then the condition number of the matrix obtained is at least $\kappa=\mo(dP^2)$. This indicates that the complexity is at least $\mo(dP^2)$ when block-encoding its inverse matrix using the previous method, such as LCU or QSVT. However, here we achieved $\tilde{\mathcal{O}}_{\epsilon}(dP)=\tilde{\mathcal{O}}_{\epsilon}(\sqrt{d\kappa})$ complexity for encoding the discretization of $(-\frac{1}{4\pi^2}\Delta + 1)^{-1}$, where we omit the logarithm terms in $\tilde{\mathcal{O}}$ and denote the dependence on $\epsilon$ by the subscript. This improvement demonstrates the potential of directly working on the symbol level, as we did in this section. When solving the corresponding problem  \eqref{eq:pdeconst} with a particular right-hand side $b$, the worst case gate complexity becomes $\tilde{\mathcal{O}}_{\epsilon}(\kappa\sqrt{d\kappa})$ since the worst case success probability is $\mo(1/\kappa)$, which is inferior comparing with the dependence on $\kappa$ obtained in \cite{costa2022optimal}. However, the block encoding scheme in this paper is independent of the right-hand side $b$ and thus and be used repeatedly for different $b$ without constructing the circuit again. Also, the circuit is simpler than the one in \cite{costa2022optimal}, making it more applicable in practice. 
\end{remark}

\section{Conclusion and Discussion}\label{sec:con}
This paper systematically investigates block encodings for pseudo-differential operators (PDOs) under different structural assumptions. A block encoding scheme for PDOs with generic symbols is developed in \Cref{sec:arithmD}, and the quantum circuit is illustrated in \Cref{fig:PDO1}. For PDOs with linear combinations of separable symbols, we improve the success probability exponentially and present an efficient block encoding algorithm in \Cref{sec:separable}. Then a more explicit and practical block encoding scheme is derived in \Cref{sec:fully sep} with the help of QSP and QET, along with which the complexity analysis is provided. Plenty of worked examples are given in \Cref{sec:app}, including the block encoding of elliptic operators with a variable coefficient that is difficult to deal with for quantum solvers that use finite difference schemes and the block encoding of the inverse of constant-coefficient elliptic operators without using quantum linear system algorithms. The block encoding schemes presented in this paper enrich the study of the block encoding of dense operators and shed new light on designing practical quantum circuits for scientific computing. 

For future directions, one can apply the established results in this paper to other PDOs besides the ones presented in \Cref{sec:app}. One can also use the idea of symbol calculus to implement different operations on the PDO, such as taking square root or exponential, which can help solve certain PDEs in practice. 

\bibliographystyle{abbrvnat}
\bibliography{ref}

\clearpage

\appendix
\section{Multiplication of two states}
\begin{prop}\label{prop:naivemult}
Let $U_a$ and $U_b$ be unitary matrices of size $N\times N$ with the first column being $[a_0, a_1, \ldots, a_{N-1}]^\top$ and $[b_0, b_1, \ldots, b_{N-1}]^\top$, respectively. Then the following circuit gives the state $\frac{1}{c}\sum_{j=0}^{N-1}a_jb_j\ket{j}$ with probability $c^2$, where $c = \sqrt{\sum_{j=0}^{N-1}|a_jb_j|^2}$. 
\end{prop}
\begin{figure}[ht]
\centerline{                                                                  
\Qcircuit @C=2em @R=1.5em {
    \lstick{\ket{0}} & \multigate{3}{U_a} & \ctrl{4} & \qw & \qw& \qw& \qw& \qw \\
    \lstick{\ket{0}} &\ghost{U_a} &\qw&\ctrl{4}&\qw& \qw& \qw& \qw\\
    \lstick{\vdots} &&&&&&\vdots&\\
    \lstick{\ket{0}} &\ghost{U_a} &\qw&\qw&\qw& \ctrl{4}& \qw& \qw\\
    \lstick{\ket{0}} & \multigate{3}{U_b} &\targ & \qw&\qw &\qw& \qw&\meter\\
    \lstick{\ket{0}} &\ghost{U_b} &\qw&\targ&\qw& \qw& \qw&\meter\\
    \lstick{\vdots} &&&&&&&\vdots\\
    \lstick{\ket{0}} &\ghost{U_b} &\qw&\qw& \qw& \targ& \qw&\meter\\
}
}    
\caption{Circuit for element-wise multiplication} \label{fig:lemma1} 
\end{figure}
\begin{proof}
After applying $U_a$ and $U_b$, the state becomes
\[\sum_{j, k=0}^{N-1}a_jb_k\ket{j}\ket{k} = \sum_{\substack{j_0, j_1, \ldots, j_{n-1}\\ k_0, k_1, \ldots, k_{n-1}}}a_jb_k\ket{j_0\cdots j_{n-1}}\ket{k_0\cdots k_{n-1}},\]
and after applying the CNOT gates, the $k$ register is only $\ket{0}$ when $j_l = k_l$ for all $l=0, \ldots, n-1$. which means the state becomes
\[\sum_{\substack{j_0, j_1, \ldots, j_{n-1}}}a_jb_j\ket{j_0\cdots j_{n-1}}\ket{0} + \ket{\perp} = \sum_{j=0}^{N-1}a_jb_j\ket{j}\ket{0} + \ket{\perp},\]
where $\ket{\perp}$ is a term that is orthogonal to $\ket{j}\ket{0}$ for any $j$. Thus the probability of obtaining $\ket{0}$ after the measurement is $\sum_{j=0}^{N-1}|a_jb_j|^2 = c^2$, and the outcome of the system register is $\frac{1}{c}\sum_{j=0}^{N-1}a_jb_j\ket{j}$ when the measurement gives $\ket{0}$. 
\vspace{3pt}
\end{proof}

\end{document}